\documentclass[11pt]{article}
\usepackage{paralist}
\usepackage{amsmath}
\usepackage{amsthm}
\usepackage{amssymb}
\usepackage{amsfonts}
\usepackage{array}
\usepackage{xspace}
\usepackage{boxedminipage}
\usepackage{enumerate}

\advance \topmargin by -1.0in
\advance \oddsidemargin by -0.8in

        {\hspace*{\fill}$\Box$\par}

\newcommand{\calH}{{\cal{H}}}
\newtheorem{theorem}{Theorem}[section]

\newtheorem{lemma}[theorem]{Lemma}
\newtheorem{corollary}[theorem]{Corollary}

\newtheorem{remark}[theorem]{Remark}

\newtheorem{claim}[theorem]{Claim}
\newtheorem{definition}[theorem]{Definition}
\newtheorem{fact}[theorem]{Fact}
\newtheorem{Definition}[theorem]{Definition}
\newtheorem{Conjecture}[theorem]{Conjecture}
\newcommand{\RR}{\mathbb R}
\newcommand{\ZZ}{\mathbb Z}

\newcommand{\partdiff}[2]{\frac{\partial {#1}}{\partial {#2}}}

\newcommand{\mixdiff}[3]{\frac{\partial^2 {#1}}{{\partial {#2}}{\partial {#3}}}}

\newcommand{\cI}{{\cal I}}
\newcommand{\cF}{{\cal F}}
\newcommand{\cG}{{\cal G}}

\newcommand{\calI}{\mathcal{I}}

\newcommand{\calP}{\mathcal{P}}

\newcommand{\calU}{\mathcal{U}}
\newcommand{\calM}{\mathcal{M}}

\newcommand{\bu}{{\boldsymbol{u}}}
\newcommand{\bv}{{\boldsymbol{v}}}

\newcommand{\bx}{{\boldsymbol{x}}}
\newcommand{\by}{{\boldsymbol{y}}}

\newcommand{\val}{{\mathbf{Val}}}
\newcommand{\Ex}{\mathop{\bf E\/}}
\newcommand{\Inf}{\mathbf{Inf}}   % functionals are either in boldface or blackboard bold.  Mainly bold for multiletter, bb for single-letter.
\newcommand{\Stab}{\mathbf{Stab}}

\newcommand{\nae}{\mathbf{NAE}}

\newcommand{\eps}{\epsilon}

\def\b1{{\bf 1}}

\def\bx{{\bf x}}
\def\by{{\bf y}}
\def\bv{{\bf v}}

\def\RR{{\mathbb R}}
\def\opt{\mathrm{OPT}}

\newcommand{\E}{\mathop{\bf E\/}}
\newcommand{\ignore}[1]{}

\newcommand{\uniquegames}{\textsc{Unique Games}\xspace}

\newcommand{\R}{\mathbb R}

\newcommand{\N}{\mathbb N}
\newcommand{\Z}{\mathbb Z}

\newcommand{\la}{\langle}
\newcommand{\ra}{\rangle}

\renewcommand{\Pr}{\mathop{\bf Pr\/}}                    % should we change these to \mathbb for consistency of single-letter functionals

\newcommand{\poly}{\mathrm{poly}}

\newcommand{\avg}{\mathop{\mathrm{avg}}}
\newcommand{\Var}{\mathop{\bf Var\/}}

\begin{document}

\title{Local Distribution and the Symmetry Gap: \\ Approximability of
Multiway Partitioning Problems\footnote{This is a full version of the paper
that appeared in ACM-SIAM SODA 2013 \cite{EVW13}.}}
\author{Alina Ene\thanks{Department of Computer Science, University of
Illinois at Urbana-Champaign. Supported in part by NSF grants
CCF-0728782 and CCF-1016684. Part of this work was done while the
author was visiting the IBM Almaden Research Center.}
 \and Jan Vondr\'ak\thanks{IBM Almaden Research Center, San Jose, CA}
 \and Yi Wu\thanks{Department of Computer Science, Purdue University.
 Part of this work was done while the author was a postdoc at the IBM
 Almaden Research Center.}}
\date{}
\maketitle

%\pagenumbering{arabic}
%\setcounter{page}{1}%Leave this line commented out.

\begin{abstract}
We study the approximability of multiway partitioning problems,
examples of which include Multiway Cut, Node-weighted Multiway Cut,
and Hypergraph Multiway Cut. We investigate these problems from the
point of view of two possible generalizations: as Min-CSPs, and as
Submodular Multiway Partition problems. These two generalizations
lead to two natural relaxations, the {\em
Basic LP}, and the {\em Lov\'asz relaxation}. The Basic LP is
generally stronger than the Lov\'asz relaxation,
but grows exponentially with the arity of the Min-CSP. The
relaxations coincide in some cases such as Multiway Cut where they
are both equivalent to the CKR relaxation.

We show that the Lov\'asz relaxation gives a $(2-2/k)$-approximation
for Submodular Multiway Partition with $k$ terminals, improving a
recent $2$-approximation \cite{CE11b}. We prove that this factor is
optimal in two senses: (1) A $(2-2/k-\epsilon)$-approximation for
Submodular Multiway Partition with $k$ terminals would require
exponentially many value queries.
% (in the oracle model), or imply $NP=RP$ (for certain explicit submodular functions)
(2) For Hypergraph Multiway Cut and Node-weighted Multiway Cut with $k$
terminals, both special cases of Submodular Multiway Partition, we
prove that a $(2-2/k-\epsilon)$-approximation is NP-hard, assuming
the Unique Games Conjecture.

Both our hardness results are more general: (1) We show that the
notion of {\em symmetry gap}, previously used for submodular
maximization problems \cite{V09,DobVon12}, also implies hardness
results for submodular minimization problems. (2) Assuming the Unique
Games Conjecture, we show that the Basic LP gives an optimal approximation
for every Min-CSP that includes the Not-Equal predicate. 

Finally, we connect the two hardness techniques by proving that the
{\em integrality gap} of the Basic LP coincides with the
{\em symmetry gap} of the multilinear relaxation (for a related
instance). This shows that the appearance of the same hardness
threshold for a Min-CSP and the related submodular minimization
problem is not a coincidence.
\end{abstract}

\newcommand{\mc}{\textsc{Graph-MC}\xspace}
\newcommand{\nwmc}{\textsc{Node-Wt-MC}\xspace}
\newcommand{\hmc}{\textsc{Hypergraph-MC}\xspace}
\newcommand{\submp}{\textsc{Sub-MP}\xspace}
\newcommand{\submpsym}{\textsc{Sub-MP-Sym}\xspace}

\section{Introduction}

In this paper, we study the approximability of {\em multiway
cut/partitioning problems}, where a ground set $V$ should be
partitioned into $k$ parts while minimizing a certain objective
function.  Classical examples of such problems are Multiway Cut (that
we abbreviate by \mc), Node-weighted Multiway Cut (\nwmc) and
Hypergraph Multiway Cut (\hmc). These problems are NP-hard but admit
constant-factor approximations.

\medskip
\noindent{\bf Multiway Cut} (\mc): {\em Given a graph $G = (V,E)$
with weights on the edges and $k$ terminals $t_1, t_2, \ldots, t_k
\in V$, remove a minimum-weight set of edges so that every two
terminals are disconnected.}

\medskip
\noindent{\bf Node-weighted Multiway Cut} (\nwmc): {\em Given a graph
$G = (V,E)$ with weights on the nodes and $k$ terminals $t_1, t_2,
\ldots, t_k \in V$, remove a minimum-weight set of vertices so that
every two terminals are disconnected.}

\medskip
\noindent{\bf Hypergraph Multiway Cut} (\hmc): {\em Given a
hypergraph $H = (V,E)$ with weights on the hyperedges and $k$
terminals $t_1, t_2, \ldots, t_k \in V$, remove a minimum-weight set
of hyperedges so that every two terminals are disconnected.}

\medskip
Although the problems above are formulated as vertex/edge/hyperedge removal
problems, \mc and \hmc can be also viewed as partitioning problems where
vertices are assigned to terminals and we pay for each edge/hyperedge
that is {\em cut} between different terminals. The \nwmc problem can
also be stated in this form, and in fact shown to be
approximation-equivalent to \hmc (although the reduction is more
complicated, see \cite{OFN10}).  Given this point of view, and the fact
that the cut function in graphs/hypergraphs is submodular, the
following generalization of these problems was proposed in
\cite{ZNI05} and recently studied in \cite{CE11a,CE11b}.

\medskip
\noindent{\bf Submodular Multiway Partition} (\submp): {\em Given a
submodular set function $f:2^V \rightarrow \RR_+$ and $k$ terminals
$t_1, t_2, \ldots, t_k \in V$, find a partition of $V$ into $A_1,
\ldots, A_k$ such that $t_i \in A_i$ and $\sum_{i=1}^{k} f(A_i)$ is
minimized.}

\medskip
This problem captures the problems \mc, \nwmc and \hmc as special
cases.\footnote{We point out that in the case of \hmc, the reduction
is not as direct as one might expect - this is due to the fact that
we want to count each cut hyperedge only once, independently of how
many terminals share it. Consequently, the arising submodular
function is {\em not} symmetric, although the cut function in a
hypergraph is.} One of the useful aspects of viewing the problems in
this more general framework is that non-trivial linear programs,
discovered on a case-by-case basis in the past, can be viewed in a
unified way: by considering the {\em Lov\'asz extension} of a
submodular function (see Section~\ref{sec:submod-MP} for details). In
particular, \cite{CE11a} rederives the geometric CKR relaxation for
Multiway Cut \cite{CKR98} in this way. It was shown in \cite{CE11b}
that \submp is not significantly more difficult than its special
cases: \submp admits a 2-approximation in general, and an improved
$(3/2 - 1/k)$-approximation when the number of terminals is $k$ and
$f$ is a symmetric submodular function (in the sense that $f(S) =
f(\bar{S})$; we denote this special case as \submpsym).  We remark
that \mc is a special case of \submpsym, while \nwmc and \hmc are
not.

This compares to the approximability of the classical problems as
follows: \mc admits a $1.2965$-approximation, which has been obtained
by a sequence of successive improvements \cite{CKR98,KKSTY99,BNS13,SV14},
all based on the CKR relaxation of \cite{CKR98}. The $(3/2-1/k)$-approximation
obtained by \cite{CKR98} matches the result of \cite{CE11b} for \submpsym.
On the hardness side, it was proved that the CKR relaxation
provides the optimal approximation factor for \mc, assuming the Unique Games
Conjecture \cite{MNRS08}. However, the actual approximation factor is
not known: it is only known that it is between $8/7$ and $1.2965$.
In the case of $3$ terminals, the optimal factor is known to be $12/11$
\cite{KKSTY99,MNRS08}.

The problems \nwmc and \hmc are known to be approximation-equivalent,
and both admit a $(2-2/k)$-approximation for $k$ terminals
\cite{OFN10,GVY04}. It is known that a $(2-\epsilon)$-approximation
independent of $k$ would imply a $(2-\epsilon)$-approximation for
Vertex Cover, which would refute the Unique Games Conjecture.
Therefore, we do not expect an approximation better than $2$ for
$\nwmc$ and $\hmc$ when the number of terminals is large.
Nevertheless, this reduction does not give any hardness for a
constant number of terminals $k$, and the optimal approximation
for a given $k$ was not known.

\medskip
\noindent{\bf Our contribution.}
We study these partitioning problems from two points of view:
(a) as general partitioning problems with a submodular cost function,
using a natural convex program for the problem based on the Lov\'asz
extension of a submodular function; we refer to this convex program
as the {\em Lov\'asz relaxation};
(b) regarding them as Min-CSPs (constraint satisfaction problems),
which leads to another natural relaxation that has been referred to
as the {\em Basic LP} (see, e.g., \cite{KOTYZ12,ThapperZ12}).\footnote{In the
ACM-SIAM SODA 2013 version of this paper, we referred to this LP as {\em Local Distribution LP}.}
Our concrete results are as follows.

\medskip
\noindent{\bf Concrete results:}
\begin{itemize}
\item We give a $(2-2/k)$-approximation for the \submp problem with
$k$ terminals using the Lov\'asz relaxation.
We also show that this is optimal, in two different senses.
\item We prove that any $(2-2/k-\epsilon)$-approximation for $\submp$
(for a constant number of terminals $k$) requires exponentially many
value queries in the oracle model.
%, and it implies $NP=RP$ (for certain succinctly represented submodular functions, using the technique of \cite{DobVon12}).
\item We prove that for \nwmc, a special case of \submp, it is
Unique-Games-hard to achieve a $(2-2/k-\epsilon)$-approximation (for a
constant number of terminals $k$).
\end{itemize}
Since $\hmc$ is approximation-equivalent to $\nwmc$, we determine the
approximability of all three problems, $\submp$, $\hmc$ and $\nwmc$,
to be exactly $2-2/k$ (assuming the Unique Games Conjecture in the
case of $\hmc$ and $\nwmc$).

%We also present a generic LP rounding algorithm that matches the
%integrality gap of the Basic LP. A similar rounding
%algorithm also applies to the {\em Lov\'asz relaxation} of \submp,
%achieving a ratio arbitrarily close to its integrality gap.  This
%rounding algorithm is inspired by the technique of
%\cite{kumar2011lp}.

\medskip
\noindent{\bf \uniquegames-hardness vs. NP-hardness more generally:}
Our hardness proofs in fact lead to more general results, revealing
an interesting relationship between the \uniquegames-hardness of Min-CSP
problems and NP-hardness of their natural submodular generalizations.
\begin{itemize}
\item We show a \uniquegames-based hardness result for Min-CSP problems,
generalizing the machinery of \cite{MNRS08}. Roughly speaking, we
show that for every Min-CSP problem that includes the Not-Equal
predicate, the integrality gap of the Basic LP can be
translated to a \uniquegames-hardness result.
\item We show how the {\em symmetry gap} technique, previously
developed for submodular maximization problems \cite{V09},
applies to submodular minimization problems. This technique yields
hardness results in the value oracle model, or computational hardness
assuming $NP \neq RP$, using the technique of \cite{DobVon12}.
In particular, we prove that it is hard to achieve a better than
$1.268$-approximation for \submpsym (a special case of \submp where
the cost function is symmetric submodular).
\end{itemize}
Finally, we present a connection between the two approaches, proving
that the {\em integrality gap} of the Basic LP coincides
with the {\em symmetry gap} of the multilinear relaxation (see the
discussion below and Section~\ref{app:integr-sym} for more details).

\medskip
\noindent{\bf Discussion.}
\ignore{The interplay of different relaxations and different techniques to prove related results is in our opinion the most interesting aspect of our work. In fact, there is a certain connection between the
dictatorship tests that appear in the \uniquegames-hardness results and
symmetric instances used in the oracle hardness results that we discuss in Section~\ref{sec:connection}.}
Let us comment on some connections that we observed here.

\noindent
{\em Integrality gap vs.~symmetry gap.}
While \uniquegames-hardness results typically start from an
integrality gap instance, hardness results for submodular functions
often start from the {\em multilinear relaxation} of a problem,
exhibiting a certain {\em symmetry gap} (see \cite{V09,DobVon12}).
This is a somewhat different concept, where instead of integral
vs.~fractional solutions, we compare symmetric vs.~asymmetric
solutions. In this paper, we clarify the relationship between the
two: For any integrality gap Min-CSP instance of the Basic LP,
there is a related Min-CSP instance that exhibits
the same symmetry gap in its multilinear relaxation. Conversely, for
any symmetry gap instance of the multilinear relaxation of a Min-CSP
instance, there is a related Min-CSP instance whose Basic LP
has the same integrality gap (see
Section~\ref{app:integr-sym}). Therefore, the two concepts are in
some sense equivalent (at least for Min-CSP problems).  This explains
why the \uniquegames-hardness threshold for \hmc and the NP-hardness
threshold for its submodular generalization \submp are the same.

\noindent
{\em Lov\'asz vs.~multilinear relaxation.}
The fact that the symmetry gap technique gives optimal results for a
submodular {\em minimization} problem is interesting: The symmetry
gap technique is intimately tied to the notion of a {\em multilinear
extension} of a submodular function, which has recently found
numerous applications in maximization of submodular functions
\cite{V08,V09,KST09,LMNS09,CVZ11}.  Nevertheless, it has been common
wisdom that the {\em Lov\'asz extension} is the appropriate extension
for submodular minimization \cite{Lovasz83,IwataN09,CE11a,CE11b}.
Here, we obtain a positive result using the Lov\'asz extension, and a
matching hardness result using the multilinear extension.

\medskip
\noindent {\bf Erratum.}
In the conference version \cite{EVW13}, it was claimed that for every fixed $k \geq 1$ and $\epsilon>0$
there is an efficient algorithm to round the Basic LP for a Min $k$-CSP instance,
achieving an approximation factor within $\epsilon$ of the integrality gap.
This result was found to contain errors and has been removed.

\medskip
\noindent {\bf Organization.}
The rest of the paper is organized as follows. In
Section~\ref{sec:submod-MP}, we discuss the Lov\'asz relaxation, and
show how it yields a $(2-2/k)$-approximation for the \submp problem.
In Section~\ref{sec:oracle-hardness}, we present the symmetry gap
technique for submodular minimization problems, and show how it
implies our hardness results in the value oracle model. In
Section~\ref{sec:min-CSP}, we present our hardness result for
Min-CSP, and show how it implies the hardness result for \hmc.  In
section~\ref{app:Lovasz-LD}, we discuss the relationship of the
Lov\'asz relaxation and the Basic LP.  In
Section~\ref{app:integr-sym}, we discuss the relationship of
integrality gaps and symmetry gaps.  

\def\SubMPRel{\textsc{SubMP-Rel}\xspace}

\section{Approximation for Submodular Multiway Partition}
\label{sec:submod-MP}

In this section, we revisit the convex relaxation proposed by Chekuri
and Ene \cite{CE11a}, and provide an improved analysis that gives the
following result.

\begin{theorem}
\label{thm:submod-MP}
There is a polynomial-time $(2-2/k)$-approximation for the \submp
problem with $k$ terminals, where $k$ and the terminals $t_1,\ldots,t_k$
are given on the input and the cost function is given by a value oracle.
\end{theorem}

\medskip
\noindent{\bf The Lov\'asz relaxation.}
The following is the convex relaxation that has been used by Chekuri
and Ene:
\begin{eqnarray*}
(\SubMPRel)
& \min \sum_{i=1}^{k} \hat{f}(\bx_i): \\
\forall j \in V; & \sum_{i=1}^{k} x_{i,j} = 1; \\
\forall i \in [k]; & x_{i,t_i} = 1; \\
\forall i,j; & x_{i,j} \geq 0.
\end{eqnarray*}

Here, $\hat{f}(\bx_i)$ denotes the Lov\'asz extension of a submodular
function.  The function $\hat{f}$ can be defined in several
equivalent ways (see \cite{CE11a,CE11b}). One definition is based on
the following rounding strategy. We choose a uniformly random $\theta
\in [0,1]$ and define $A_i(\theta) = \{ j: x_{ij} > \theta \}$. Then
$\hat{f}(\bx_i) = \E[f(A_i(\theta))]$.  Equivalently (for submodular
functions), $\hat{f}$ is the convex closure of $f$ on $[0,1]^V$. The
second definition shows that the relaxation \SubMPRel is a convex
program and therefore it can be solved in polynomial time.

Given a fractional solution, we use the following randomized rounding
technique, a slight modification of one proposed by Chekuri and Ene:

\medskip
\noindent{\bf Randomized rounding for the Lov\'asz relaxation.}
\begin{itemize}
\item Choose $\theta \in (\frac12, 1]$ uniformly at random and define
$A_i(\theta) = \{j: x_{ij} > \theta \}$.
\item Define $U(\theta) = V \setminus \bigcup_{i=1}^{k} A_i(\theta) =
\{j: \max_i x_{ij} \leq \theta \}$.
\item Allocate each $A_i(\theta)$ to terminal $i$, and in addition
allocate $U(\theta)$ to a terminal $i'$ chosen uniformly at random.
\end{itemize}
Each terminal $t_i$ is allocated to itself with probability one.
Moreover, the sets $A_i(\theta)$ are disjoint by construction, and
therefore the rounding constructs a feasible solution. The only
difference from Chekuri and Ene's rounding \cite{CE11b} is that we
assign the ``unallocated set" $U(\theta)$ to a random terminal rather
than a fixed terminal.  (However, taking advantage of this in the
analysis is not straightforward.) We prove the following.

\begin{theorem}
\label{thm:2-2/k}
The above rounding gives a feasible solution of expected value at most
$(2 - \frac{2}{k}) \sum_{i=1}^{k} \hat{f}(\bx_i)$.
\end{theorem}

This implies Theorem~\ref{thm:submod-MP}. In the following,
we prove Theorem~\ref{thm:2-2/k}.  We assume that $f(\emptyset) = 0$.
This is without loss of generality, as the value of the empty set can
be decreased without violating submodularity and this does not affect
the problem (since terminals are always assigned to themselves).

We start by defining several sets, parameterized by $\theta$,
that will be important in the analysis.

\begin{itemize}
\item $A_i(\theta) = \{j: x_{ij} > \theta \}$
\item $A(\theta) = \bigcup_{i=1}^{k} A_i(\theta) = \{j: \max_i x_{ij} > \theta \}$
\item $U(\theta) = V \setminus A(\theta) =  \{j: \max_i x_{ij} \leq \theta \}$.
\item $B(\theta) = U(1-\theta) = \{j: 1-\max_i x_{ij} \geq \theta \}$.
\end{itemize}

We can express the LP cost and the cost of the rounded solution in
terms of these sets as follows. The following lemma follows immediately from the definition of the Lov\'asz extension.

\begin{lemma}
\label{lem:LP}
The cost of the LP solution is
$$ LP = \sum_{i=1}^{k} \int_0^1 f(A_i(\theta)) d\theta.$$
\end{lemma}

\noindent
The next lemma gives an expression for the expected value achieved by the algorithm in a form convenient for the analysis.

\begin{lemma}
\label{lem:ALG}
The expected cost of the rounded solution is
	$$ALG = \left(2-\frac{2}{k} \right) \sum_{i=1}^{k} \int_{1/2}^{1}
	f(A_i(\theta)) d\theta + \frac{2}{k} \sum_{i=1}^{k}
	\int_{0}^{1/2} f(A_i(\theta) \cup B(\theta)) d\theta.$$
\end{lemma}

\begin{proof}
The set allocated to terminal $i$ is $A_i(\theta)$ with probability
$1-1/k$, and $A_i(\theta) \cup U(\theta)$ with probability $1/k$. We
are choosing $\theta$ uniformly between $\frac12$ and $1$. This gives
the expression
	$$ALG = \left(2-\frac{2}{k} \right) \sum_{i=1}^{k} \int_{1/2}^{1}
	f(A_i(\theta)) d\theta + \frac{2}{k} \sum_{i=1}^{k}
	\int_{1/2}^{1} f(A_i(\theta) \cup U(\theta)) d\theta.$$
We claim that for $\theta \in [\frac12,1]$, $A_i(\theta) \cup
U(\theta)$ can be written equivalently as $A_i(1-\theta) \cup
B(1-\theta)$. We consider three cases for each element $j$:
\begin{itemize}
\item If $x_{ij} > \frac12$, then $j \in A_i(\theta) \cup U(\theta)$ for every $\theta \in [\frac12,1]$, because $x_{i'j} < \frac12$ for every other $i' \neq i$ and hence $j$ cannot be allocated to any other terminal. Similarly, $j \in A_i(1-\theta) \cup B(1-\theta)$ for every $\theta \in [\frac12,1]$, because $1-\theta \leq \frac12$ and so $j \in A_i(1-\theta)$.
\item If $x_{ij} \leq \frac12$ and $x_{ij} = \max_{i'} x_{i'j}$, then again $j \in A_i(\theta) \cup U(\theta)$ for every $\theta \in [\frac12,1]$, because $j$ is always in the unallocated set $U(\theta)$. Also, $j \in A_i(1-\theta) \cup B(1-\theta)$, because $B(1-\theta) = U(\theta)$.
\item If $x_{ij} \leq \frac12$ and $x_{ij} < \max_{i'} x_{i'j}$, then $j \in A_i(\theta) \cup U(\theta)$ if and only if $j \in U(\theta) = B(1-\theta)$. Also, we have $x_{ij} = 1 - \sum_{i'\neq i} x_{i'j} \leq 1 - \max_{i'} x_{i'j}$, and therefore $j \in A_i(1-\theta) \cup B(1-\theta)$ if and only if $j \in B(1-\theta)$.
\end{itemize}

To summarize, for every $\theta \in [\frac12,1]$, $j \in A_i(\theta) \cup U(\theta)$ if and only if
$j \in A_i(1-\theta) \cup B(1-\theta)$. 
Therefore, the total expected cost can be written as
\begin{eqnarray*}
 ALG & = & \left(2-\frac{2}{k} \right) \sum_{i=1}^{k} \int_{1/2}^{1}
 f(A_i(\theta)) d\theta + \frac{2}{k} \sum_{i=1}^{k} \int_{1/2}^{1}
 f(A_i(\theta) \cup U(\theta)) d\theta \\
 & = & \left(2-\frac{2}{k} \right) \sum_{i=1}^{k} \int_{1/2}^{1}
 f(A_i(\theta)) d\theta + \frac{2}{k} \sum_{i=1}^{k} \int_{1/2}^{1}
 f(A_i(1-\theta) \cup B(1-\theta)) d\theta \\
 & = & \left(2-\frac{2}{k} \right) \sum_{i=1}^{k} \int_{1/2}^{1}
 f(A_i(\theta)) d\theta + \frac{2}{k} \sum_{i=1}^{k} \int_{0}^{1/2}
 f(A_i(\theta) \cup B(\theta)) d\theta.
\end{eqnarray*}
\end{proof}

In the rest of the analysis, we prove several inequalities that
relate the LP cost to the ALG cost.  Note that the integrals
$\int_{1/2}^{1} f(A_i(\theta)) d\theta$ appear in both LP and ALG.
The non-trivial part is how to relate $\int_{0}^{1/2} f(A_i(\theta))
d\theta$ to $\int_{0}^{1/2} f(A_i(\theta) \cup B(\theta)) d\theta$. 

The following statement was proved in \cite{CE11b}; we give a
simplified new proof in the process of our analysis.

\begin{lemma}[Theorem~{1.5} in \cite{CE11b}]
\label{lem:CE-lemma}
	Let $f \geq 0$ be submodular, $f(\emptyset) = 0$, and $\bx$ a
	feasible solution to \SubMPRel.
	For $\theta \in [0,1]$ let $A_i(\theta) = \{ v \mid x_{v,i} >
	\theta \}$, $A(\theta) = \cup_{i=1}^k A_i(\theta)$ and $U(\theta)
	= V \setminus A(\theta)$.  For any $\delta \in [\frac12, 1]$ the
	following holds:
		$$\sum_{i = 1}^k \int_0^{\delta} f(A_i(\theta)) d\theta \geq
		\int_0^{\delta} f(A(\theta))d\theta + \int_0^1
		f(U(\theta))d\theta.$$
\end{lemma}

In the following, we assume the conditions of
Lemma~\ref{lem:CE-lemma} without repeating them.
First, we prove the following inequality.
% (to prove Lemma~\ref{lem:CE-lemma} as well as another lemma necessary for our improved analysis).

\begin{lemma}
\label{lem:cupcap}
For any $\delta \in [\frac12, 1]$,
                $$\sum_{i = 1}^{k - 1} \int_0^\delta f((A_1(\theta)
                \cup \cdots A_i(\theta)) \cap A_{i + 1}(\theta)) d\theta
                \geq \int_0^1 f(U(\theta)) d\theta.$$
\end{lemma}

\begin{proof}
%		In the following, we assume without loss of generality that
%		the function $f$ satisfies\footnote{If $f(\emptyset) > 0$,
%		we simply change $f(\emptyset)$ to zero. The
%		resulting function remains submodular 
%		and it does not affect the lemma, as the sets on the left-hand side
%		are always nonempty
%		(the terminal $t_i$ is always in $A_i(\theta)$).}
%		$f(\emptyset) = 0$.

 % 	When
 %       $\delta = 1$, the inequality can be written as follows:
 %               $$\sum_{i = 1}^k \Ex_{\theta \in [0, 1]}[f(A_i(\theta))] \geq
  %              \Ex_{\theta \in [0, 1]}[f(A(\theta))] + \Ex_{\theta \in [0,
  %              1]}[f(U(\theta))].$$

	First consider $\delta=1$.
	We can view the value $\int_0^1 f(A_1(\theta) \cup \cdots \cup A_i(\theta)) \cap A_{i + 1}(\theta)) d\theta$ as the Lov\'asz extension evaluated on the vector $\by_i = (\bx_1 \vee \cdots \vee \bx_i) \wedge \bx_{i+1}$.
	%For each $v \in V$
        %and each $i \neq k$, let $y_{v,i} = \min\{\max\{x_{v,1}, \cdots,
        %x_{v, i} \}, x_{v, i+1} )\}$.
	 Note that $v \in (A_1(\theta) \cup
        \cdots \cup A_i(\theta)) \cap A_{i+1}(\theta)$ if and only if
        $y_{v,i} \geq \theta$.  Therefore
                $$\int_0^1 f((A_1(\theta) \cup \cdots \cup
                A_i(\theta)) \cap A_{i+1}(\theta)) d\theta = \hat{f}(\by_i).$$
        We can also view $f(U(\theta))$ as follows:
	Let $\bu = \sum_{i = 1}^{k - 1} \by_i = \b1 - (\bx_1 \vee \cdots \vee \bx_k)$.
	(This holds because $\sum_{i=1}^{k-1} \by_i + (\bx_1 \vee \cdots \vee \bx_k)  = \sum_{i=1}^{k-1} ((\bx_1 \vee \cdots \vee \bx_i) \wedge \bx_{i+1}) + (\bx_1 \vee \cdots \vee \bx_k)  = \sum_{i=1}^{k} \bx_i$,
	which can be proved by repeated use of the rule $(\bu \wedge \bv) + (\bu \vee \bv) = \bu + \bv$,
	and finally $\sum_{i=1}^{k} \bx_i = \b1$.)
	Therefore
        \begin{align*}
                {1 \over k - 1} \sum_{i = 1}^{k - 1} \hat{f}(\by_i) &\geq
                \hat{f}\left({1 \over k - 1} \sum_{i = 1}^{k - 1}
                \by_i\right) \qquad \mbox{($\hat{f}$ is convex)}\\
                &= \hat{f}\left({1 \over k - 1} \bu\right) = {1 \over k - 1}
                \hat{f}(\bu)
        \end{align*}
	where we also used the fact that $\hat{f}(\alpha \bx) = \alpha \hat{f}(\bx)$ for any $\alpha \in [0,1]$
	($\hat{f}(\bx)$ is linear under multiplication by a scalar).
	Equivalently,
                $$\sum_{i = 1}^{k - 1} \int_0^1 f((A_1(\theta)
                \cup \cdots A_i(\theta)) \cap A_{i + 1}(\theta)) d\theta 
                 \geq \int_0^1 f(U(\theta)) d\theta.$$

         Now note that, if $\theta > \delta \geq 1/2$, the sets $(A_1(\theta) \cup
         \cdots \cup A_i(\theta)) \cap A_{i + 1}(\theta)$ are empty, since
	$\sum_{i=1}^{k} \bx_i = \b1$ and hence
	two vectors $\bx_j, \bx_{i+1}$ cannot have the same coordinate
	larger than $\frac12$. We also assumed that $f(\emptyset) = 0$,
	so we proved in fact
                $$\sum_{i = 1}^{k-1} \int_0^\delta f((A_1(\theta)
                \cup \cdots A_i(\theta)) \cap A_{i + 1}(\theta)) d\theta 
                 \geq \int_0^1 f(U(\theta)) d\theta $$
	as desired.
\end{proof}

\medskip\noindent
Given this inequality, Lemma~\ref{lem:CE-lemma} follows easily:

\begin{proof}[Lemma~\ref{lem:CE-lemma}]
	By applying submodularity inductively to the sets $A_1(\theta)
	\cup \cdots \cup A_i(\theta)$ and $A_{i+1}(\theta)$, we get
	\begin{eqnarray*}
		\sum_{i = 1}^k f(A_i(\theta)) & \geq & \sum_{i = 1}^{k-1}
		f((A_1(\theta) \cup \cdots \cup A_i(\theta)) \cap A_{i +
		1}(\theta)) + f(A_1(\theta) \cup \cdots \cup A_k(\theta)) \\
		& = & \sum_{i = 1}^{k-1}  f((A_1(\theta) \cup \cdots \cup
		A_i(\theta)) \cap A_{i + 1}(\theta)) + f(A(\theta)).
	\end{eqnarray*}
	Integrating from $0$ to $\delta$ and using Lemma~\ref{lem:cupcap}, we obtain
	\begin{eqnarray*}
		\sum_{i=1}^{k} \int_0^\delta f(A_i(\theta)) d\theta & \geq &
		\sum_{i = 1}^{k-1}  \int_0^\delta f((A_1(\theta) \cup \cdots
		\cup A_i(\theta)) \cap A_{i + 1}(\theta)) d\theta
        + \int_0^\delta f(A(\theta)) d\theta \\
		& \geq & \int_0^1 f(U(\theta)) d\theta + \int_0^\delta A(\theta) d\theta.
	\end{eqnarray*}
\end{proof}

\medskip\noindent
A corollary of Lemma~\ref{lem:CE-lemma} is the following inequality.

\begin{lemma}
\label{lem:bound-1}
	$$\sum_{i=1}^{k} \int_{0}^{1/2} f(A_i(\theta)) d\theta \geq
	\int_{0}^{1/2} f(B(\theta)) d\theta.$$
\end{lemma}
\begin{proof}
	Considering Lemma~\ref{lem:CE-lemma}, we simply note that
	$U(\theta) = B(1-\theta)$. We discard the contribution of
	$f(A(\theta))$ and keep only one half of the integral involving
	$B(1-\theta)$.
\end{proof}

\medskip\noindent
We combine this bound with the following lemma.

\begin{lemma}
\label{lem:bound-2}
	$$\sum_{i=1}^{k} \int_{0}^{1/2} f(A_i(\theta)) d\theta \geq
	\sum_{i=1}^{k} \int_{0}^{1/2} f(A_i(\theta) \cup B(\theta))
	d\theta - (k-2) \int_{0}^{1/2} f(B(\theta)) d\theta.$$
\end{lemma}

\begin{proof}
	For simplicity of notation, we drop the explicit dependence on
	$\theta$, keeping in mind that all the sets depend on $\theta$.
	By submodularity, we have $f(A_i) + f(B) \geq f(A_i \cup B) +
	f(A_i \cap B)$.  Therefore,
	\begin{eqnarray*}
		\sum_{i=1}^{k} f(A_i) &\geq& \sum_{i=1}^{k} (f(A_i \cup B) +
		f(A_i \cap B) - f(B))\\
		&=&  \sum_{i=1}^{k} f(A_i \cup B) + \sum_{i=1}^{k} f(A_i \cap
		B) - k \cdot f(B)
	\end{eqnarray*}
This would already prove the lemma with $k$ instead of $k-2$;
however, we use $\sum_{i=1}^{k} f(A_i \cap B)$ to save the additional
terms. We apply a sequence of inequalities using submodularity,
starting with $f(A_1 \cap B) + f(A_2 \cap B) \geq f(A_1 \cap A_2 \cap
B) + f((A_1 \cup A_2) \cap B)$, then $f((A_1 \cup A_2) \cap B) +
f(A_3 \cap B) \geq f((A_1 \cup A_2) \cap A_3 \cap B) + f((A_1 \cup
A_2 \cup A_3) \cap B)$, etc. until we obtain
	$$\sum_{i=1}^{k} f(A_i \cap B) \geq \sum_{i=1}^{k-1} f((A_1
	\cup \ldots \cup A_i) \cap A_{i+1} \cap B) + f((A_1 \cup \ldots
	\cup A_k) \cap B).$$
The last term is equal to $f(A \cap B)$. Moreover, we observe that
for every element $j$, at most one variable $x_{ij}$ can be larger
than $1 - \max_{i'} x_{i'j}$ (because otherwise the two variables
would sum up to more than $1$). Therefore for every $i$, $(A_1 \cup
\ldots \cup A_i) \cap A_{i+1} \subseteq B$. So we get
	$$\sum_{i=1}^{k} f(A_i \cap B) \geq \sum_{i=1}^{k-1} f((A_1 \cup
	\ldots \cup A_i) \cap A_{i+1}) + f(A \cap B).$$
Integrating from $0$ to $1/2$, we get
	$$\sum_{i=1}^{k} \int_{0}^{1/2} f(A_i \cap B) d\theta \\ \geq
	\sum_{i=1}^{k-1} \int_{0}^{1/2}  f((A_1 \cup \ldots \cup A_i)
	\cap A_{i+1}) d\theta + \int_{0}^{1/2} f(A \cap B)
	d\theta.$$
By Lemma~\ref{lem:cupcap} (recalling that $A_i = A_i(\theta)$), we obtain
	$$\sum_{i=1}^{k} \int_{0}^{1/2} f(A_i \cap B) d\theta \geq
	\int_{0}^{1}  f(U) d\theta +  \int_{0}^{1/2} f(A \cap B)
	d\theta.$$
Using $B(\theta) = U(1-\theta)$, submodularity, and the fact that $U$
is the complement of $A$, we obtain
\begin{eqnarray*}
	\sum_{i=1}^{k} \int_{0}^{1/2} f(A_i \cap B) d\theta & \geq &
	\int_{0}^{1/2} f(B) d\theta + \int_{0}^{1/2}  f(U) d\theta +
	\int_{0}^{1/2}  f(A \cap B) d\theta  \\
	& \geq & \int_{0}^{1/2} f(B) d\theta + \int_{0}^{1/2} f(U \cup (A
	\cap B)) d\theta \\
	& = & \int_{0}^{1/2} f(B) d\theta + \int_{0}^{1/2} f(U \cup B)
	d\theta
\end{eqnarray*}
Finally, for $\theta \in [0,\frac12]$, we claim that $U \cup B = B$.
This is because if $\max_i x_{ij} > \frac12$, then $j \notin U$, and
hence the membership on both sides depends only on $j \in B$. If
$\max_i x_{ij} \leq \frac12$, then $j \in B$ and hence also $j \in U
\cup B$.  We conclude that
	$$\sum_{i=1}^{k} \int_{0}^{1/2} f(A_i \cap B) d\theta \geq 2
	\int_{0}^{1/2} f(B) d\theta$$
and 
\begin{eqnarray*}
	\sum_{i=1}^{k} \int_{0}^{1/2} f(A_i) d\theta & \geq &
	\sum_{i=1}^{k} \int_{0}^{1/2} (f(A_i \cup B) + f(A_i \cap B) -
	f(B)) d\theta \\
	& \geq & \sum_{i=1}^{k} \int_{0}^{1/2} f(A_i \cup B) d\theta -
	(k-2) \int_{0}^{1/2} f(B) d\theta
\end{eqnarray*}
which finishes the proof.
\end{proof}

\medskip\noindent
A combination of Lemma~\ref{lem:bound-1} and Lemma~\ref{lem:bound-2}
relates $\sum_{i=1}^{k} \int_{0}^{1/2} f(A_i(\theta)) d\theta$ to
$\sum_{i=1}^{k} \int_{0}^{1/2} f(A_i(\theta) \cup B(\theta))
d\theta$, and finishes the analysis.

\begin{proof}[Proof of Theorem~\ref{thm:2-2/k}]
Add up $\frac{k-2}{k-1} \times$ Lemma~\ref{lem:bound-1} $+
\frac{1}{k-1} \times$ Lemma~\ref{lem:bound-2}:
	$$ \sum_{i=1}^{k} \int_{0}^{1/2} f(A_i(\theta)) d\theta \geq
	\frac{1}{k-1} \sum_{i=1}^{k} \int_{0}^{1/2} f(A_i(\theta) \cup
	B(\theta)) d\theta.$$
Adding $ \sum_{i=1}^{k} \int_{1/2}^{1} f(A_i(\theta)) d\theta$ to both
sides gives us that
	$$\sum_{i=1}^{k} \int_{0}^{1} f(A_i(\theta)) d\theta \geq
	\sum_{i=1}^{k} \int_{1/2}^{1} f(A_i(\theta)) d\theta +
	\frac{1}{k-1} \cdot \sum_{i=1}^{k} \int_{0}^{1/2}
	f(A_i(\theta) \cup B(\theta)) d\theta.$$
The left-hand side is equal to $LP$, while the right-hand side is
equal to $\frac{ALG}{2-2/k}$ (see Lemma~\ref{lem:ALG}).
\end{proof}

\section{Hardness from the Symmetry Gap}
\label{sec:oracle-hardness}

Here we show how the symmetry gap technique of \cite{V09full} applies to
submodular minimization problems.  We remark that while the technique
was presented in \cite{V09full} for submodular maximization problems, it
applies to submodular minimization problems practically without any
change.  Rather than repeating the entire construction of \cite{V09full},
we summarize the main components of the proof and point out the
important differences. Finally, we mention that the recent techniques
of \cite{DobVon12} turn a query-complexity hardness result into a
computational hardness result.  First, we show the result for general
Submodular Multiway Partition, which is technically simpler.

\subsection{Hardness of \submp}

Here we show that the $(2-2/k)$-approximation is optimal for
Submodular Multiway Partition in the value oracle model. More
precisely, we prove the following.

\begin{theorem}
\label{thm:SMP-hardness}
	For any fixed $k>2$ and $\epsilon>0$, a
	$(2-2/k-\epsilon)$-approximation for the Submodular Multiway
	Partition problem with $k$ terminals in the value oracle model
	requires exponentially many value queries.
	% (or it implies that $NP	= RP$ for certain succinctly represented submodular functions).
\end{theorem}

We note that this result can also be converted into a computational hardness result
for explicit instances, using the techniques of \cite{DobVon12}.
We defer the details to Appendix~\ref{sec:oracle-to-NP}
and focus here on the value oracle model.

A starting point of the hardness construction is a particular
instance of the problem which exhibits a certain {\em symmetry gap},
a gap between symmetric and asymmetric solutions of the multilinear
relaxation.  We propose the following instance (which is somewhat
related to the gadget used in \cite{DahlhausJPSY92} to prove the
APX-hardness of Multiway Cut).  The instance is in fact an instance
of Hypergraph Multiway Cut (which is a special case of Submodular
Multiway Partition). As in other cases, we should keep in mind that
this does not mean that we prove a hardness result for Hypergraph
Multiway Cut, since the instance gets modified in the process.

\medskip
\noindent{\bf The symmetric instance.}
Let the vertex set be $V = [k] \times [k]$ and let the terminals be
$t_i = (i,i)$.  We consider $2k$ hyperedges: the ``rows" $R_i = \{
(i,j): 1 \leq j \leq k\}$ and the ``columns" $C_j = \{ (i,j): 1 \leq
i \leq k \}$.  The submodular function $f:2^V \rightarrow \RR_+$ is
the following function: for each set $S$, $f(S) = \sum_{i=1}^{k}
\phi(|S \cap R_i|) + \sum_{j=1}^{k} \phi(|S \cap C_j|)$, where
$\phi(t) = t/k$ if $t<k$ and $\phi(t) = 0$ if $t=k$.

\

Since $\phi$ is a concave function, it follows easily that $f$ is
submodular.  Further, if a hyperedge is assigned completely to one
terminal, it does not contribute to the objective function, while if
it is partitioned among different terminals, it contributes $t/k$ to
each terminal containing $t$ of its $k$ vertices, and hence $1$
altogether.  Therefore, $\sum_{i=1}^{k} f(S_i)$ captures exactly the
number of hyperedges cut by a partition $(S_1, \ldots, S_k)$.

\medskip
\noindent{\bf The multilinear relaxation.}
Along the lines of \cite{V09full}, we want to compare symmetric and
asymmetric solutions of the {\em multilinear relaxation} of the
problem, where we allocate vertices fractionally and the objective
function $f:2^V \rightarrow \RR_+$ is replaced by its multilinear
extension $F:[0,1]^V \rightarrow \RR_+$; we have $F(\bx) =
\E[f(\hat{\bx})]$, where $\hat{\bx}$ is the integral vector obtained
from $\bx$ by rounding each coordinate independently.  The
multilinear relaxation of the problem has variables $x^{\ell}_{ij}$
corresponding to allocating $(i,j)$ to terminal $t_\ell$.
\begin{eqnarray*}
\min \Big\{ \sum_{\ell=1}^{k} F(\bx^{\ell}) & : & \forall i,j \in
[k]; \sum_{\ell=1}^{k} x^{\ell}_{ij} = 1, \\
 & & \forall i \in [k]; x^{i}_{ii} = 1, \\
 & & \forall i,j,\ell \in [k]; x^{\ell}_{ij} \geq 0 \Big\}.
\end{eqnarray*}
In fact, this formulation is equivalent to the discrete problem,
since any fractional solution can be rounded by assigning each vertex
$(i,j)$ independently with probabilities $x^{\ell}_{ij}$, and the
expected cost of this solution is by definition $\sum_{\ell=1}^{k}
F(\bx^{\ell})$.

\medskip
\noindent{\bf Computing the symmetry gap.}
What is the symmetry gap of this instance?  It is quite easy to see
that there is a symmetry between the rows and the columns, i.e., we
can exchange the role of rows and columns and the instance remains
the same.  Formally, the instance is invariant under a group $\cG$ of
permutations of $V$, where $\cG$ consists of the identity and the
transposition of rows and columns.  A symmetric solution is one
invariant under this transposition, i.e., such that the vertices
$(i,j)$ and $(j,i)$ are allocated in the same manner, or
$x^{\ell}_{ij} = x^{\ell}_{ji}$. For a fractional solution $\bx$, we
define the symmetrized solution as $\bar{\bx} = \frac12 (\bx +
\bx^T)$ where $\bx^T$ is the transposed solution $(x^T)^{\ell}_{ij} =
x^{\ell}_{ji}$.

There are two optimal solutions to this problem: one that assigns
vertices based on rows, and one that assigns vertices based on
columns. The first one can be written as follows: $x^{\ell}_{ij} = 1$
iff $i = \ell$ and $0$ otherwise.  (One can recognize this as a
``dictator" function, one that copies the first coordinate.) The cost
of this solution is $k$, because we cut all the column hyperedges and
none of the rows.  We must cut at least $k$ hyperedges,
because for any $i \neq j$, we must cut either row $R_i$ or column
$C_j$. Since we can partition all hyperedges into pairs like this
($\{R_1,C_2\}$, $\{R_2,C_3\}$, $\{R_3,C_4\}$, etc.), at least a half
of all hyperedges must be cut. Therefore, $OPT = k$.

Next, we want to find the optimal symmetric solution. As we observed,
there is a symmetry between rows and columns and hence we want to
consider only solutions satisfying $x^{\ell}_{ij} = x^{\ell}_{ji}$
for all $i,j$.  Again, we claim that it is enough to consider integer
(symmetric) solutions. This is for the following reason: we can
assign each pair of vertices $(i,j)$ and $(j,i)$ in a coordinated
fashion to the same random terminal: we assign $(i, j)$ and $(j, i)$
to the terminal $t_\ell$ with probability $x^{\ell}_{ij} =
x^{\ell}_{ji}$. Since these two vertices never participate in the
same hyperedge, the expected cost of this correlated randomized
rounding is equal to the cost of independent randomized rounding,
where each vertex is assigned independently.  Hence the expected
cost of the rounded symmetric solution is exactly $\sum_{\ell=1}^{k}
F(\bx^\ell)$.

Considering integer symmetric instances yields the following optimal
solution: We can assign all vertices (except the terminals
themselves) to the same terminal, let's say $t_1$. This will cut all
hyperedges except $2$ (the row $R_1$ and the column $C_1$). This is
in fact the minimum-cost symmetric solution, because once we have any
monochromatic row (where monochromatic means assigned to the same
terminal), the respective column is also monochromatic. But this row
and column intersect all other rows and columns, and hence no other
row or column can be monochromatic (recall that the terminals are on
the diagonal and by definition are assigned to themselves).  Hence, a
symmetric solution can have at most 2 hyperedges that are not cut.
Therefore, the symmetric optimum is $\overline{OPT} = 2k-2$ and the
symmetry gap is $\gamma = (2k-2)/k = 2 - 2/k$.

\medskip
\noindent{\bf The hardness proof.}
We appeal now to a technical lemma from \cite{V09full}, which serves
to produce blown-up instances from the initial symmetric instance.

\begin{lemma}
\label{lemma:final-fix}
Consider a function $f:2^V \rightarrow \RR$ that is invariant under a
group of permutations $\cal G$ on the ground set $X$.  Let $F(\bx) =
\E[f(\hat{\bx})]$, $\bar{\bx} = \E_{\sigma \in \cG}[\sigma(\bx)]$,
and fix any $\epsilon > 0$.  Then there exists $\delta > 0$ and functions
$\hat{F}, \hat{G}:[0,1]^V \rightarrow \RR_+$ (which are also
symmetric with respect to $\cG$), satisfying:
\begin{enumerate}
\item For all $\bx \in [0,1]^V$, $\hat{G}(\bx) = \hat{F}(\bar{\bx})$.
\item For all $\bx \in [0,1]^V$, $|\hat{F}(\bx) - F(\bx)| \leq \epsilon$.
\item Whenever $||\bx - \bar{\bx}||^2 \leq \delta$, $\hat{F}(\bx) = \hat{G}(\bx)$
% = \hat{F}(\bar{\bx})$.
and the value depends only on $\bar{\bx}$.
%\item $\max \{\hat{F}(\bx): \bx \in P({\cF})\} \geq OPT$.
%\item $\max \{\hat{G}(\bx): \bx \in P({\cF})\} \leq (1+\epsilon) \overline{OPT}$.
\item The first partial derivatives of $\hat{F}, \hat{G}$ are absolutely continuous.\footnote{
A function $F:[0,1]^V \rightarrow \RR$ is
absolutely continuous, if $\forall \epsilon>0; \exists \delta>0;
\sum_{i=1}^{t} ||\bx_i-\by_i|| < \delta \Rightarrow \sum_{i=1}^{t}
|F(\bx_i) - F(\by_i)| < \epsilon$.}  
\item If $f$ is monotone, then $\partdiff{\hat{F}}{x_i} \geq 0$ and
   $\partdiff{\hat{G}}{x_i} \geq 0$ everywhere.
\item If $f$ is submodular, then $\mixdiff{\hat{F}}{x_i}{x_j} \leq 0$ and
   $\mixdiff{\hat{G}}{x_i}{x_j} \leq 0$ almost everywhere.
\end{enumerate}
\end{lemma}

We apply Lemma~\ref{lemma:final-fix} to the function $f$ from the
symmetric instance.  This will produce continuous functions $\hat{F},
\hat{G}: 2^V \rightarrow \RR_+$.  Next, we use the following lemma
from \cite{V09full} to discretize the continuous function $\hat{F},
\hat{G}$ and obtain instances of \submp.

\begin{lemma}
\label{lemma:smooth-submodular}
Let $F:[0,1]^V \rightarrow \RR$ be a function with absolutely
continuous first partial derivatives. Let $N
= [n]$, $n \geq 1$, and define $f:N \times V \rightarrow \RR$ so that
$f(S) = F(\bx)$ where $x_i = \frac{1}{n} |S \cap (N \times \{i\})|$.
Then
\begin{enumerate}
\item If $\partdiff{F}{x_i} \geq 0$ everywhere for each $i$, then $f$
is monotone.
\item If $\mixdiff{F}{x_i}{x_j} \leq 0$ almost everywhere for all
$i,j$, then $f$ is submodular.
\end{enumerate}
\end{lemma}

Using Lemma~\ref{lemma:smooth-submodular}, we define blown-up
instances on a ground set $X = N \times V$ as follows: For each $i
\in N$, choose independently a random permutation $\sigma \in \cG$ on
$V$, which is either the identity or the transposition of rows and
columns.  Then for a set $S \subseteq N \times V$, we define $\xi(S)
\in [0,1]^V$ as follows:
	$$ \xi_j(S) = \frac{1}{n} \left|\{i \in N: (i,\sigma^{(i)}(j))
	\in S \}\right|.$$
We define two functions
$\hat{f},\hat{g}:2^V \rightarrow \RR_+$, where
	$$ \hat{f}(S) = \hat{F}(\xi(S)), \ \ \ \ \ \hat{g}(S) =
	\hat{G}(\xi(S)).$$
By Lemma~\ref{lemma:smooth-submodular}, $\hat{f}, \hat{g}$ are
submodular functions.
We consider the following instances of \submp:
	$$\max \Bigg\{ \sum_{\ell=1}^{k} \hat{f}(S_\ell):
	(S_1,\ldots,S_k) \mbox{ is a partition of }X
	\ \& \ \forall i \in N; \forall \ell \in [k]; (i,t_\ell) \in
	S_\ell \Bigg\},$$
	$$\max \Bigg\{ \sum_{\ell=1}^{k} \hat{g}(S_\ell):
	(S_1,\ldots,S_k) \mbox{ is a partition of }X
	\ \& \ \forall i \in N; \forall \ell \in [k]; (i,t_\ell) \in
	S_\ell \Bigg\}.$$
Note that in these instances, multiple vertices are required to be assigned
to a certain terminal ($n$ vertices for each terminal).
However, this can be still viewed as a Submodular Multiway Partition
problem; if desired, the set of pre-labeled vertices $T_\ell = N
\times \{t_\ell\}$ for each terminal can be contracted into one
vertex.

Finally, we appeal to the following lemma in \cite{V09full}.

\begin{lemma}
\label{lemma:indistinguish}
Let $\hat{F}, \hat{G}$ be the two functions provided by
Lemma~\ref{lemma:final-fix}.  For a parameter $n \in \ZZ_+$ and $N =
[n]$, define two discrete functions $\hat{f}, \hat{g}: 2^{N \times V}
\rightarrow \RR_+$ as follows: Let $\sigma^{(i)}$ be an arbitrary
permutation in $\cG$ for each $i \in N$.  For every set $S \subseteq
N \times V$, we define a vector $\xi(S) \in [0,1]^V$ by
	$$ \xi_j(S) = \frac{1}{n} \left|\{i \in N: (i,\sigma^{(i)}(j))
	\in S \}\right|.$$
Let us define:
$ \hat{f}(S) = \hat{F}(\xi(S))$, $\hat{g}(S) = \hat{G}(\xi(S)).$
Then deciding whether a function given by a value oracle is $\hat{f}$
or $\hat{g}$ (even using a randomized algorithm with a constant
probability of success) requires an exponential number of queries.
\end{lemma}

Lemma~\ref{lemma:indistinguish} implies that distinguishing these
pairs of objective functions requires an exponential number of
queries. We need to make one additional argument, that the knowledge
of the terminal sets $T_\ell = N \times \{ t_\ell \}$ (which is part
of the instance) does not help in distinguishing the two objective
functions. This is because given oracle access to $\hat{f}$ or
$\hat{g}$, we are in fact able to identify the sets $T_\ell$, if we
just modify the contribution of each row/column pair $R_\ell, C_\ell$
by a factor of $1+\epsilon_\ell$, where $\epsilon_\ell$ is some
arbitrary small parameter. This does not change the optimal values
significantly, but it allows an algorithm to distinguish the sets
$T_\ell$ easily by checking marginal values. Note that then we can
also determine sets such as $T_{i,j} = N \times \{ (i,j), (j,i) \}$,
but we cannot distinguish the two symmetric parts of $T_{i,j}$, which
is the point of the symmetry argument. In summary, revealing
the sets $T_\ell$ does not give any information that the algorithm
cannot determine from the value oracles for $\hat{f}, \hat{g}$, and
given this oracle access, $\hat{f}$ and $\hat{g}$ cannot be
distinguished.

It remains to compare the optima of the two optimization problems.
The problem with the objective function $\hat{f}$ corresponds to the
multilinear relaxation with objective $\hat{F}$, and admits the
``dictatorship" solution $S_i = \{(i,j): 1 \leq j \leq k\}$ for each
$i \in [k]$, which has a value close to $k$.  On the other hand, any
solution of the problem with objective function $\hat{g}$ corresponds
to a fractional solution of the symmetrized multilinear relaxation of
the problem with objective $\hat{G}$, which as we argued has a value
close to $2k-2$.  Therefore, achieving a
$(2-2/k-\epsilon)$-approximation for any fixed $\epsilon>0$ requires
an exponential number of value queries.

%%%%%%%%%%%%%%%%%%%%%%%%

\subsection{Hardness of Symmetric Submodular Multiway Partition}

Here we state a result for the \submpsym problem.

\begin{theorem}
\label{thm:sym-SMP-oracle-hardness}
	For any fixed $k$ sufficienly large, a better than $1.268$-approximation for
	the \submpsym problem with $k$ terminals requires exponentially
	many value queries.
	% (or it implies that $NP=RP$ for certain 	succinctly represented submodular functions).
\end{theorem}

The proof is
essentially identical to the previous section, however the symmetric
instance is different due to the requirement that the submodular
function itself be symmetric (in the sense that $f(S) = f(\bar{S})$).
The analysis of the symmetry gap in this case is technically more
involved than in the previous section.  The result that we obtain is
as follows; we defer the proof to Appendix~\ref{app:oracle-hardness}.

\section{Hardness from Unique Games}
\label{sec:min-CSP}

In this section, we formulate our general hardness result for Min-CSP
problems, and in particular we show how it implies the hardness
result for Hypergraph Multiway Cut (\hmc).

\subsection{Min-CSP and the Basic LP}
\label{sec:Min-CSP}

The Min-CSPs we consider consist of a set of variables and a set of
predicates (or cost functions) with constant arity over the
variables. The goal is to assign a value from some finite domain to
each variable so as to minimize the total cost of an assignment.
Alternatively, we can view these variables as vertices of a
hypergraph and the predicates being evaluated on the hyperedges of
the hypergraph.

\begin{Definition} Let $\ss=\{ \Psi :[q]^i\to [0,1]\cup \{\infty\} \ |\ i\leq k\}$ be a collection of functions with each function in $\ss$ has at most $k$ input  variables in $[q]$ and outputs a value in $[0,1]$. We call $k$ the arity  and $q$ the alphabet size of the $\ss$.

An instance of the Min-$\ss$-CSP,  specified by $\left(V,E,\Psi_E=\{\Psi_e\ |\ e\in E\}, w_E=\{w_e |\ \in E\}\right)$, is defined over a weighted $k$-multi-hypergraph \(G(V,E)\).   For every hyperedge  $e=(v_{i_1},v_{i_2},..,v_{i_j})\in E$, there is an
associated cost function $\Psi_e\in \ss \) and a
positive weight $w_e$. The goal is find an assignment $\ell: V\to [q]$
for each vertex $v \in V$ so as to minimize
\[
\sum_{e=(v_{i_1},\ldots,v_{i_j})\in E}w_e\cdot \Psi_e(\ell(v_{i_1}),\ldots,\ell(v_{i_{j}})).
\]

% \begin{itemize}
% \item $V=\{x_1,x_2, \ldots,x_n\}$ where $x_1, x_2, \ldots, x_n$ are variables taking values in $[q]$.
% \item for each $e\in E$, there is a cost function $\Psi_e$ applied on it.
% \item there is a positive weight $w_e$ associated with each edge $e$.
% \end{itemize}
If there is a subset of vertices (which are called the terminal vertices) such that each has a single required label, we call the corresponding problem Min-$\ss$-TCSP. If for every vertex $v$, it is only allowed to choose a label from a candidate list $L_v\subseteq [q]$ ($L_v$ is also part of the input), we call the corresponding problem Min-$\ss$-LCSP.

\end{Definition}

Given a Min-$\ss$-CSP (or  Min-$\ss$-TCSP/Min-$\ss$-LCSP) instance, it is natural to write down the following linear
program. We remark that this LP can be seen as a generalization of the Earthmover LP from \cite{MNRS08},
and has been referred to as the {\em Basic LP} \cite{KOTYZ12,ThapperZ12}. The LP captures probability distributions
over the possible assignments to each constraint (which is why we referred to this as the {\em Local Distribution LP} in the conference version of this paper).

\medskip
\noindent{\bf The Basic LP.}
There are variables $x_{e,\alpha}$ for every hyperedge $e \in E$ and
assignment $\alpha \in [q]^{|e|}$, and variables $x_{v,j}$ for every
vertex $v\in V$ and $j\in [q]$.  The objective function is of the
following form:

\[
 LP(\calI) = \min\sum_{e} w_e\cdot \sum
 _{\alpha\in [q]^{|e|}}x_{e,\alpha}\cdot \Psi_e(\alpha)
\]
under the constraint that, for every $v\in V$,
\[
\sum_{i=1}^q x_{v,i} =1
\]
where $0\leq x_{v,i}\leq 1$ for every $v,i$. We also have constraints
that for every edge \(e=(v_1,v_2,\ldots, v_{j})\) and \(i\in
S$ and $q_0\in [q]$
\[x_{v_i,q_0}=\sum_{\alpha_i=q_0} x_{e,\alpha}.\]
where $ 0\leq x_{v_i,\alpha}\leq 1$  for every $x_{v_i,\alpha}$.

As for Min-$\ss$-LCSP as well as  Min-$\ss$-TCSP, we would add the following additional constraint:
 for every  $q_0$  that  is not a feasible label assignment for $v$, we would add $$x_{v,q_0}=0.$$

To see why this is a relaxation, one should think of $x_{v,i}$ as the
probability of assigning label $i$ to vertex $v$. For every edge
$e=(v_1,v_2,\ldots,v_j)$ and labeling $\alpha = (l_1,l_2,\ldots,l_j)$
for the vertices of $e$, $x_{e,\alpha}$ is the probability of
labeling the vertices of $e$ according to $\alpha$ (that is, vertex
$v_i$ receives label $l_i$). For every edge $e$, we define $\calP_e$
as the distribution that assigns probability $x_{e,\alpha}$ to
each $\alpha\in [q]^{|e|}$.

\subsection{The Min-CSP hardness theorem}

\begin{definition}For some $i\geq 2$,we define $\nae_i(x_1, x_2, ...x_i):[q]^i \to {0,1}$ to be $0$ if $x_1=x_2,\ldots,=x_i$  and $1$ otherwise.
\end{definition}

\begin{theorem}\label{thm:ughardness}
Suppose we have a Min-$\ss$-CSP(TCSP/LCSP) instance $\calI(V,E,\Psi_E,w_E)$ with fractional optimum (of the Basic LP) $LP(\calI)=c$, integral optimum $\opt(\calI)=s$, and $\ss$ contains the predicate $\nae_2$. Then for any $\eps > 0$, for some $\lambda>0$ it is \uniquegames-hard to distinguish between instances of Min-$\ss$-CSP(TCSP/LCSP) where the optimum value is at least $(s-\eps) \lambda$, and instances where the optimum value is less than $(c+\eps) \lambda$.
\end{theorem}

As a corollary of Theorem~\ref{thm:ughardness},
we obtain a hardness result for Hypergraph Multiway Cut.  This
follows from a known integrality gap example for Hypergraph
Multiway Cut, reformulated for the Basic LP.

\begin{corollary}\label{cor:hmcg}
        The \hmc problem with $k$ terminals is \uniquegames-hard to
        approximate within $(2-\frac{2}{k}-\epsilon)$ for any fixed
        $\epsilon>0$. The same hardness result holds even the hyperedge of the graph has size at most $k$.
\end{corollary}
\begin{proof}Let $\ss_k=\{\nae_i:[k]^i\to \{0,1\} \ |\ i=2\ldots, k\}$.

First, we claim if we have an $\alpha$-approximation for the \hmc with $k$ terminals for some constant $\alpha>1$, then we can also have an $\alpha$-approximation for the Min-$\ss_k$-TCSP. To see this, we make the following reduction. Take any instance of the Min-$\ss_k$-TCSP instance, it can  almost be viewed as a $k$-way $\hmc$ instance on $k$-hypergraph as each constraint $\nae_i(v_1,v_2, \ldots,v_i)$ is corresponding to a hyperedge on $v_1,v_2, \ldots, v_i$.  The only difference that is there may be multiple  vertices fixed  to be  the same label in the Min-$\ss_k$-TCSP instance. To address this, we only need to add $k$ new terminals $t_1,t_2, \ldots, t_k$. For all the existing vertex associated with the label $i$ in the Min-$\ss_k$-CSP instance, we would add an edge of infinite weight to the corresponding $t_i$.

Therefore, it remains to show the hardness of approximating Min-$\ss_k$-CSP better than $2-\frac{2}{k}$. Assuming the correctness of Theorem~\ref{thm:ughardness}, consider the following Min-$\ss_k$-CSP instance $\calH_k$:  there are $k(k+1)/2$ vertices indexed by $(i,j)$ for $1\leq i\leq j\leq k$. We have $k$ hyperedges: for every $i \in [k]$, the hyperedge $e_i$ is defined as $e_i = \{ (i_1,i_2) \in [k]^2: i = i_1 \leq i_2 \mbox{ or } i_1 \leq i_2 =i \}$. We define the $k$ terminals as $t_i = (i,i)$ for every $i \in [k]$, with the label of $t_i$ required to be $i$.

We claim that $\opt(\calH_k)\geq k-1$; i.e., there is no assignment with cost $0$ on more than one edge. Without loss of generality, suppose the optimal solution has cost  $0$  on edge $e_1$; i.e, assign label $1$ to every vertex indexed by $(1,i)$ for $i\in [k]$. Then we cannot have cost $0$ for any of the remaining $k-1$ hyperedges because to satisfy $e_i$, we would need  $(1,i)$ to be labeled by $i$.

On the other hand, $LP(\calH_k) \leq k/2.$ The following is a fractional solution:  for every vertex $v=(i,j)$, $x_{v,i}=1/2$ and $x_{v,j}=1/2$. All the other variables $x_{v,k'}$ are $0$ (for $k'\ne i,j$).  For every edge $e_i$ with its vertices ordered as $(1,i),(2,i),\ldots,(i,i),(i,i+1),\ldots(i,k)$, we have $x_{e,(i,i,\ldots,i)}=1/2$ and $x_{e,(1,2,\ldots,k)}=1/2$.  This satisfies all the constraints and achieves an objective value of $k/2$.

% Now we apply Theorem~\ref{thm:ughardness} with $\calI = \calH_k$. We need to verify that we indeed get a hardness result for a $\hmc$ instance with $k$ terminals. The hard instances arising from Theorem~\ref{thm:ughardness} have the same type of constraints as $\hmc$; the $\nae_2$ predicate is already in the predicate function class of \hmc.
Therefore, applying Theorem~\ref{thm:ughardness}, we get that it is \uniquegames-hard to approximate Min-$\ss_k$-TCSP  beyond the factor $\frac{k-1}{k/2}=2-\frac{2}{k}$, which implies the same hardness of approximation ratio for the $k$-way  \hmc problem (even on $k$-hypergraph as the arity of $\ss_k$ is $k$).
\end{proof}

In the following, we give a proof of Theorem~\ref{thm:ughardness}, which is by an extension of the technique of \cite{MNRS08}.
In Section~\ref{sec:harmonic}, we first review some standard definitions from the analysis of boolean functions.
Then we describe our reduction and analyze it in Section~\ref{sec:reduction}.

\subsection{Tools from Discrete Harmonic Analysis }
\label{sec:harmonic}

We now recall some standard definitions from the analysis of boolean functions. We will be considering functions of the form $f : [q]^n \to \R^k$, where $q, n, k \in \N$.  The set of all functions $f : [q]^n \to \R^k$ forms a vector space with inner product
\[
\la f, g \ra = \Ex_{x \sim [q]^n} [\la f(x), g(x)\ra];
\]
here we mean that $x$ is uniformly random and the $\la \cdot, \cdot \ra$ inside the expectation is the usual inner product in $\R^k$.  We also write $\|f\| = \sqrt{\la f, f \ra}$ as usual.
\begin{Definition}
For random $x,y\in [q]^n$, we say that $y$ is $\rho$-correlated with $x$ if given $x$, we generate $y$ by setting $y_i=x_i$ with probability $\rho$ or randomly in $[q]$ with probability $1-\rho$, independently for each $i$.

For $0 \leq \rho \leq 1$, we define $T_\rho$ to be the linear operator given by
\[
T_\rho f(x) = \Ex_{y}[f(y)],
\]
where $y$ is a random string in $[q]^n$ which is $\rho$-correlated to $x$. We define the \emph{noise stability of $f$ at $\rho$} to be
\[
\Stab_\rho[f] = \la f, T_\rho f\ra.
\]
\end{Definition}
\begin{Definition}
For $i \in [n]$, we define the \emph{influence of $i$ on $f : [q]^n \to \R^k$} to be
\[
\Inf_i[f] = \Ex_{x_1, \dots, x_{i-1}, x_{i+1}, \dots, x_n \sim [q]} \left[\Var_{x_i \sim [q]} [f(x)] \right],
\]
where $\Var[f]$ is defined to be $\E[\|f\|^2] - \|\E[f]\|^2$.  More generally, for $0 \leq \delta \leq 1$ we define the \emph{$\delta$-noisy-influence of $i$ on $f$} to be
\[
\Inf^{(1-\delta)}_i[f] = \Inf_i[T_{1-\delta} f].
\]
\end{Definition}
The following facts are well known in the literature of discrete harmonic analysis.
\begin{fact}
\[
\Inf^{(1-\delta)}_i[f] = \sum_{j = 1}^k \Inf^{(1-\delta)}_i[f_j],
\]
where $f_j : [q]^n \to \R$ denotes the $j$th-coordinate output function of $f$.
\end{fact}

\begin{fact}  \label{prop:convexity-noisy-influences}
Let $f^{(1)}, \dots, f^{(t)}$ be a collection of functions $[q]^n \to \R^k$.  For $c_1,c_2,...c_t \in \R$ (or $[q]^n \rightarrow \R^k$), we use the notation $\avg(c_1,\dots,c_t) $ to denote their (pointwise) average $\frac{1}{t} \sum_{j=1}^t c_j.$
Then
\[
\Inf^{(1-\delta)}_i\left[\avg_{j \in [t]} \left\{f^{(j)}\right\}\right] \leq \avg_{j \in [t]} \left\{\Inf^{(1-\delta)}_i[f^{(j)}]\right\}.
\]
\end{fact}

For randomized functions with discrete output, $f':[q]^n\to [k]$, we can view them as functions defined as $f:[q]^n\to \Delta_k$ where $\Delta_k$ is the $(k-1)$-dimensional standard simplex. The $i$-th coordinate indicates the probability that the function $f'$ outputs $i$. The following fact is also well known.

\begin{fact}\label{lem:ti}For any $f:[q]^n\to \Delta_k$,  $\sum_{i=1}^n \Inf_{i}^{1-\delta}[f] \leq 1/\delta.$
\end{fact}

An important tool we need is the Majority Is Stablest Theorem from~\cite{MOO05}.
We state here a slightly modified  version~\cite{Rag09} using a small
``noisy-influence'' assumption rather than a small ``low degree influence" assumption.

\begin{theorem}\label{thm:MIST}(Majority Is Stablest)
Suppose $f : [q]^n \rightarrow [0, 1]$ has  $\Inf_i^{(1-\delta)}[f] \leq \tau$ and $\E[f] = \mu$, then
\[
\Stab_{1 - \delta}[f] \leq \Gamma_{1 - \delta}(\mu) + err(\tau,q,\delta)
\]
where for any fixed $\delta$ and $q$, $\lim_{\tau \rightarrow 0} err(\tau,q,\delta) = 0$.
Here the quantity $\Gamma_{1-\delta}(\mu)$ is defined to be $\Pr[x, y \leq t]$ when $(x, y)$ are joint standard Gaussian with covariance $1-\delta$ and $t$ is defined by $\Pr[x \leq t] = \mu$.
\end{theorem}

We will use the following asymptotic estimate for $\Gamma_{1-\delta}$.

\begin{lemma}
\label{lem:tech}
If $\delta^c \leq  \mu \leq 1-\delta^c$ for some constant  $0<c<1$, $$\mu -\Gamma_{1-\delta}(\mu) = \Omega(\delta^{\frac{1}{2}+2c}).$$
\end{lemma}

\begin{proof}
Suppose that we have \(\Pr_{x\sim N(0,1)}[x\leq t] =\mu \). Since $(1-\delta)$-correlated Gaussian variables $(x,y)$ can be generated by starting with two independent Gaussian variables $(x,z)$, and setting $y = (1-\delta)x + \sqrt{2\delta-\delta^2} z$, we can write
\begin{equation}\label{eqn:ineq}
\Gamma_{1-\delta}(\mu)=\Pr[x\leq t,y\leq t] = \Pr[(1-\delta)x +\sqrt{2\delta-\delta^2}z\leq t \ | \ x\leq t] \cdot \mu
\end{equation}
where $x,z$ are independent Gaussian. For $t'=t- \sqrt{2\delta-\delta^2}$, we have that
\begin{eqnarray*}\label{eqn:cal}
\Pr[(1-\delta)x +\sqrt{2\delta-\delta^2}z\leq t \ | \ x\leq t]
& = & \Pr[x > t'|\ x\leq t] \cdot \Pr[(1-\delta)x +\sqrt{2\delta-\delta^2}z\leq t\ |\  t' < x\leq t] \\
& & +\Pr[x \leq t'\ |\ x\leq t] \cdot\Pr[(1-\delta)x +\sqrt{2\delta-\delta^2}z\leq t \ | \ x\leq t'] \\
& \leq &  \Pr[x >  t'|\ x\leq t] \cdot \Pr[z \leq \frac{t-(1-\delta)t'}{\sqrt{2\delta-\delta^2}}]+\Pr[x\leq t'|x\leq t] \\
%& = &  \Pr[x >  t'|\ x\leq t] \cdot \Pr[z \leq \frac{\delta t + (1-\delta) \sqrt{2\delta-\delta^2}}{\sqrt{2\delta-\delta^2}}]+\Pr[x\leq t'|x\leq t] \\
&=&  \Pr[x >  t'|\ x\leq t] \cdot \Pr[z \leq \frac{\delta t}{\sqrt{2\delta-\delta^2}}+(1-\delta)]+\Pr[x\leq t'|x\leq t]
\\
& \leq & \Pr[x >  t'|\ x\leq t] \cdot \Pr[z \leq \sqrt{\delta}t+(1-\delta)]+\Pr[x\leq t'|x\leq t] \\
& = & p_1 \cdot (1-p_2 )+(1-p_1)=1-p_{1}p_{2 }
\end{eqnarray*}
where $p_1=\Pr[x >  t'|\ x\leq t]$ and $p_2 =\Pr[z > \sqrt{\delta}t+(1-\delta)]$. Hence we get that
\[
        \mu -\Gamma_{1-\delta}(\mu) =\mu p_1p_2.\]
 Below we prove \begin{enumerate}
\item $p_1 =\Omega(\delta^{\frac{1}{2} +c}).$ \item $p_2 =\Omega(1)$.
 \end{enumerate}
 Combining with the fact that $\mu\geq \delta^c$,  this will  complete the proof of Lemma~\ref{lem:tech}.

We need following property of Gaussian Distribution proved in~\cite{CMM06}.
\begin{lemma}~\cite{CMM06}
Let $f(x)\ = \frac{1}{\sqrt{2\pi}}e^{-\frac{x^2}{2}}$ be the density function of Gaussian Distribution.
 For every $s>0$,
  \[
\frac{s f(s)}{s^2+1}\leq \Pr_{g\sim N(0,1)}[g\geq s] = \Pr_{g\sim N(0,1)}[g \leq -s] \leq \frac{f(s)}{s}.
\]
\end{lemma}

\medskip

Recall that $\Pr_{g \sim N(0,1)}[g \leq t] = \mu$.  In the case of $\mu \leq 1/2$ (and therefore $t<0$), we have that
\[
        \mu =\Pr_{g\sim N(0,1)}(g\leq t)\leq \frac{f(t)}{|t|}.
\]
In the case that $\mu \geq 1/2$ (and therefore, $t>0$), we also have that $$\Pr_{g\sim N(0,1)}( g\geq t)=1-\mu\leq \frac{f(t)}{|t|}$$
Combining the above two cases, we have that
\[
        \delta^c\leq \min (1-\mu,\mu)\leq \frac{f(t)}{|t|}=\frac{1}{\sqrt{2\pi} |t|}e^{-\frac{t^2}{2}}.
\]
  For sufficiently small $\delta$, we claim that    $$|t| \leq \sqrt{2c\log (1/\delta)},$$  This can be proved by contradiction. Suppose that $|t|> \sqrt{2c \log (1/\delta)}$, we would have    $$\frac{f(t)}{|t|}\leq \frac{\delta^c}{\sqrt{4\pi c \log(1/\delta)}}<\delta^c$$ when $\delta$  is small enough.
Therefore,
$$p_2 =\Pr[z > \sqrt{\delta}t+(1-\delta)] \geq \Pr[z > \sqrt{2c \delta \log (1/\delta}) + 1].$$
For sufficiently small $\delta$, we have $\sqrt{2c \delta \log (1/\delta})\leq 1$ and  $p_2 \geq   \Pr[z\geq 2]\geq 0.1.$

Next, we establish a lower bound for  $p_1$:
\begin{multline*}p_1 = \Pr[t'<x<t |x<t] \geq \Pr[t'\leq x\leq t] \geq f(\max(|t'|,|t|))\cdot(t-t') \\ \geq f(|t|+ \sqrt{2\delta-\delta^2}) \sqrt{2\delta-\delta^2} \geq \sqrt{\delta}\cdot f(|t|+\sqrt{2\delta-\delta^2}).
\end{multline*}
Noticing that $\sqrt{2\delta-\delta^2}\leq \sqrt{2\delta}$, we have
\begin{multline*}
f(|t|+ \sqrt{2\delta-\delta^2})\geq f(|t|+ \sqrt{2\delta})\geq
 \frac{{e^{-(\sqrt{2c\log(1/\delta)}+\sqrt{2\delta})^2/2}}}{\sqrt{2\pi}}=\frac{e^{-(c\log(1/\delta)+\delta+\sqrt{4\delta c \log(1/\delta)}) }}{\sqrt{2\pi}}\geq  \frac{\delta^c}{2\pi}.
\end{multline*} The last inequality holds for sufficiently small constant $\delta$ and therefore $p_1\geq \frac{\delta^{1/2+c}}{2\pi}$.
Overall, we have that $$\mu-T_{1-\delta}(\mu)\geq \mu p_1 p_2 = \Omega(\delta^{\frac{1}{2}+2c})$$ and this finishes the proof of Lemma~\ref{lem:tech}.
\end{proof}

\subsection{Reduction from Unique Games}
\label{sec:reduction}

We describe a reduction from the \uniquegames problem to Min-$\ss$-CSP
where $\ss$ contains the $\nae_2$ predicate.

\begin{Definition}
\label{def:ug}
A \uniquegames instance
$\mathcal{U}(V,E,\{\pi_{u,v}\}_{(u,v) \in E},R)$ consists of a
regular  graph $G(V,E)$ and each edge $e = (u,v) \in E$ is associated with a
permutation  $\pi_{uv} : [R] \mapsto [R]$.   For a given
labeling $L: V\mapsto [R]$, an edge $e = (u,v) \in E$ is said to be
satisfied if $L(v) = \pi_{u,v}(L(u))$.  We denote $\opt(\calU)$ to be the maximum fraction of edges that can be satisfied by all the labeling.
\end{Definition}

\begin{Conjecture} [The Unique Games Conjecture]
Given a \uniquegames instance $\mathcal{U}$,  for every $\eps >0$, there exists some large enough integer $R$ ,  given a \uniquegames instance
$\mathcal{U} \left(V,E,\{\pi_{u,v}\}_{(u,v) \in E}, R\right)$, it
is NP-hard to distinguish between the following two cases:
\begin{itemize}
\item $OPT(\mathcal{U}) \geq 1 - \eps$;
\item $OPT(\mathcal{U}) \leq \eps$.
\end{itemize}
\end{Conjecture}

Now we are ready to prove Theorem~\ref{thm:ughardness}.
We first prove the hardness result for Min-$\ss$-CSP and then extend it to  Min-$\ss$-LCSP as well as  Min-$\ss$-TCSP. The reduction takes a \uniquegames instance $\calU(V,E,$
$\{\pi_e|e\in E\},R)$ and it maps $\calU$ to a $\ss$-CSP instance $\calM$, using an
integrality gap instance in the process. Suppose the integrality gap instance is
$$\calI'= \left(V',E',\Psi_{E'}=\{\Psi_e\ |\ e\in E'\}, w_{E'}=\{w_e |\ \in E'\}\right)$$ and suppose that $|V'|\leq m$,
$\opt(\calI') = c$ and $LP(\calI')=s$. Let us assume that the arity of $\ss$ is $k$ and the alphabet size of $\ss$ is $q$.
Without loss of generality, let us assume that
the sum of the weights over all the hyperedges in $E'$ normalized to $1$;
thus we can view the weights $w_e$ as a probability distribution over hyperedges.
The reduction produces a new
$\ss$-CSP instance \(\calM\) with the following properties, for some parameter $0<\eta<1$ to be specified later:
\begin{enumerate}
\item Completeness property: if \(\opt(\calU)\geq 1-\eps\), then
$\opt(\calM)\leq  \eta(c+O(1/m))$;
\item Soundness property: if $\opt(\calU)\leq \eps$, then
\(\opt(\calM)\geq \eta(s -O(k/m))$.
\end{enumerate} 

By taking sufficiently large $m$, we have that it is \uniquegames-hard to get any approximation better than $\frac{s}{c}-\eps_0$ for any $\eps_0>0$ for the Min-$\ss$-CSP.

The vertex set of \(\calM\)'s is going to be \(V\times V'\times [q]^R\).  . The sum of the
hyperedge weights in $\calM$ is equal to one and we described it as a distribution over all the hyperedges.

Recall that we denote by $\calP_e$ the probability distribution over
assignments in $[q]^e$ corresponding to the optimal fractional solution
of $\calI'$, restricted to the hyperedge $e$. Given any $\pi:[R]\to [R]$ and $x\in [q]^R$,  we use $\pi(x)$ to
indicate a vector in $[q]^R$ such that $\pi(x)_i = x_{\pi(i)}$.
\\

\begin{boxedminipage}[c]{0.9\textwidth}
\label{fig:reduction}
\begin{center}Reduction from \uniquegames to $\ss$-CSP.
\end{center}

We choose the  parameters as follows: $m > \max(|V'|, q^{\frac{1}{4}},k^2)$, $\eps\leq 1/m^{80}, \eta=1/m^{39},\delta=1/m^{40}$. 
The weight $w_e$ of a hyperedge $e$ with cost function $\Psi_e$ is the probability that it is generated by the following procedure.

\begin{itemize}
\item (Edge test) With probability $\eta$, we pick an edge
\(e=(v'_1,v'_2,\ldots,v'_j)\) from $E'$. Then we randomly pick a
vertex $v$ from $V$ and randomly pick \(j\) of its neighbors
$v_1,v_2,\ldots,v_j$.  We generate \(x^1,x^2\ldots, x^j\in [q]^R\)
according to \(\calP_e^R\).
Output the cost function $\Psi_e$ on 
\((v_1,v'_{1},\pi_{v_1,v}(x^1)),(v_2,v'_2,\pi_{v_2,v}(x^2))\) $\ldots,(v_j,v'_j,\pi_{v_j,v}(x^j))$.

\item (Vertex test) With probability $(1-\eta)$, we pick a vertex $v$
from $V$ and two of its neighbors $v_1,v_2$. We randomly pick a vertex
\(v'\) from \( V'\). Then we generate $(1-\delta)$-correlated $x,y \in [q]^R$
and output a cost function $\nae_2$ on $(v_1,v',\pi_{v_1,v}(x))$
and $(v_2,v',\pi_{v_2,v}(y))$.
\end{itemize}
\end{boxedminipage}

\medskip\noindent
A function \(f:V\times V'\times[q]^R\to [q]\) corresponds to a labeling
of the instance $\calM$. Let us use $\val(f)$ to denote the expected cost
of $f$ and  $\val_{edge}(f)$ and $\val_{vertex}(f)$ to denote the expected cost
of $f$ on the edge test and the vertex test, respectively. Also let us use the
notation $f_{v,v'}:[q]^R\to [q]$ to denote the restriction of $f$ to
a fixed pair $v \in V, v'\in V'$: $f_{v,v'}(x) = f(v,v',x)$.
We know that $\val(f) =(1-\eta)\cdot \val_{vertex}(f)+\eta \cdot
\val_{edge}(f)$. In the following we prove the completeness and soundness property
of the reduction.

\begin{lemma}[completeness]
\label{lem:completeness}
If $\opt(\calU)\geq 1-\eps$, then  $\opt(\calM)\leq  \eta(c+O(1/m)) = c/m^{39}+O(1/m^{40})$.
\end{lemma}

\begin{proof}
Suppose that a labeling $\Lambda:V\to [R]$ satisfies
$(1-\eps)$-fraction of the edges in the \uniquegames instance.
Let us consider a ``dictator labeling" $f_{v,v'}(x) = x_{\Lambda(v)}$
for every $(v,v') \in V \times V'$.  Since a $(1-\eps)$-fraction of the edges  in the \uniquegames instance can be
satisfied, by an averaging argument and the regularity of the graph,
we know that for at least a $(1-\sqrt{\eps})$-fraction of the vertices,
we have that at least a $(1-\sqrt{\eps})$-fraction of its neighbors is satisfied
by $\Lambda$. Therefore, by a union bound, when choosing a random pair of neighbors
of a random vertex in $V$, with probability at least $1-3\sqrt{\eps}$ the two edges
are satisfied by $\Lambda$. This means that for the vertex test, with probability at
least $1-3\sqrt{\eps}$ we have $\pi_{v_1,v}(\Lambda(v_1))=
\pi_{v_2,v}(\Lambda(v_2))$. Conditioned on this, the cost of the
vertex test is at most $\delta$, as $\nae_2(\pi_{v_1,v}(x_{(\Lambda(v_1)})=\pi_{v_2,v}(y_{(\Lambda(v_2))}))=0$ with probability $(1-\delta)$.
Overall, we have that the vertex test cost is at most $(1-3\sqrt{\eps}) \cdot \delta +
3\sqrt{\eps}$.

As for the edge test, by an extension of the argument above, with probability at least
$1-(m+1)\sqrt{\eps}$, we have $\pi_{v_1,v}(\Lambda(v_1)) =
\pi_{v_2,v}(\Lambda(v_2)) \ldots = \pi_{v_j,v}(\Lambda(v_j))$.
Conditioned on this, since
$(x^1_{\pi_{v_1,v}(\Lambda(v_1))},\ldots,x^j_{\pi_{v_j,v}(\Lambda(v_j))}) \sim
\calP_e$, the cost of the edge test corresponds exactly to the cost of the fractional solution of the LP:
        $$\E_{e,v_1,v_2,\ldots,v_j}[\Psi_e(x^1_{\pi_{v_1,v}(\Lambda(v_1))},\ldots,x^j_{\pi_{v_j,v}(\Lambda(v_k)})]
        = \sum_{e}w_e\sum_{\alpha \sim \calP_e} \E[\Psi_e(\alpha)] = c.$$
Therefore, the cost of the edge test is at most
$\left(1-(m+1)\sqrt{\eps}\right)\cdot c + (m+1)\sqrt{\eps}$.

Overall, the cost of the dictator labeling is
        $$\eta\cdot \left((m+1)\sqrt{\eps}+(1-(m+1)\sqrt{\eps})\cdot
        c\right) + (1-\eta)\cdot ((1-3\sqrt{\eps})\cdot \delta +
        3\sqrt{\eps}).$$
By the choice of parameters, we obtain that $\val(f)= c/m^{39}+ O(1/m^{40}).$
\end{proof}

\medskip\noindent
It remains to prove the following soundness property.

\begin{lemma}[soundness]
\label{lem:soundness}
If $\opt(\calU)\leq \eps$, then  $\opt(\calM)\geq  \eta(s-O(k/m)) = s/m^{39}-O(k/m^{40})$.
\end{lemma}
%The proof of above lemma is more involved. Below is an overview of our proof strategy.
%\paragraph{Overview of the proof} We consider a labeling of $\cal M$ described by the labeling functions $f_{v,v'}: [q]^R \rightarrow [q]$ for each $v \in V, v' \in V'$.
%\begin{itemize}
%\item
% We classify all the functions $f_{v,v'}$ into three categories: (i)
% dictator function; (ii) constant function; (iii) low influence
% non-constant function.
%\item In order to achieve a low cost, most $f_{v,v'}$ functions can not be in category (iii).  Otherwise,
%the cost of the vertex test is overwhelming, due to
%Theorem~\ref{thm:MIST} (Majority Is Stablest).
%\item If, for some $v'$, most of the $f_{v,v'}$ functions are close to a
%dictator, then we can assign a label to each $v$ from the influential
%coordinates of $f_{v,v'}$. Such an assignment  gives a good solution
%to the \uniquegames which leads to a contradiction. Therefore the
%fraction of $f_{v,v'}$ functions that are close to a dictator cannot
%be too large.
%\item In conclusion, most $f_{v,v'}$ functions are close to constant
%functions, and $f$ can be interpreted as an integral solution of the
%original instance. Therefore the total cost corresponds to the
%integral optimum of the original instance.
%\end{itemize}

%\subsection{Proof of Lemma~\ref{lem:soundness}}
%\label{sec:soundness}

\begin{proof}
We will prove this by contradiction. Assume that there is an assignment $f$ such that $\val(f)\leq s/m^{39}-O(1/m^{40})$.
The cost of the vertex test is:
\begin{eqnarray*}
%\label{eqn:sd}
        \val_{vertex}(f) &=& \E_{v_1,v_2,v', x,y}
        [\nae_2(f_{v_1,v'}(\pi_{v_1,v}(x)),
        f_{v_{2},v'}(\pi_{v_2,v}(y)))]
              \\
        &=&\E_{v'}[\E_{v_1,v_2,x,y}
        [1- \sum_{i=1}^q  f_{v_1,v'}^i(\pi_{v_1,v}(x))\cdot f_{v_{2},v'}^i(\pi_{v_2,v}(
        y))]]     
        \\
        &=& 1-\sum_{i=1}^q \E_{v'}[\E_{v_1,v_2,x,y}
        [f_{v_1,v'}^i(\pi_{v_1,v}(x))\cdot f_{v_{2},v'}^i(\pi_{v_2,v}(
        y))]] 
        \\
        &=& 1-\sum_{i=1}^q \E_{v'}[\E_{v,x,y} [\E_{v_1\sim
        v}[f_{v_1,v'}^i(\pi_{v_1,v}(x))] \E_{v_2\sim
        v}[f_{v_{2},v'}^i(\pi_{v_2,v}(y))]]].
\end{eqnarray*}
In the above expression, $f_{v,v'}^i$ is the indicator function of whether
$f_{v,v'} =i$. Also $v_i\sim v$  means $v_i$ is a random neighbor of $v$
and $x,y \in [q]^R$ are $(1-\delta)$-correlated.
If we define $g_{v,v'}^i(x) = \E_{u\sim v }[f_{u,v'}^i(\pi_{u,v}(x)]$,
we have that 
\begin{eqnarray*}
        \val_{vertex}(f) &=& 1-\sum_{i=1}^q \E_{v'}\left[\E_{v,x,y}
        \big[g_{v,v'}^i(x)\cdot g_{v,v'}^i(y))\big] \right]\\
        &=& \E_{v,v'}\left[1-\sum_{i=1}^q\Stab_{1-\delta}[g_{v,v'}^i] \right].
\end{eqnarray*}

\noindent
Recall that $k\leq \sqrt{m}$ by the choice of parameter,  we have that   $$\val(f)=(1-\eta) \cdot \val_{vertex}(f) +\eta \val_{edge}(f) =s/m^{39}-O(k/m^{40})= O(1/m^{39}),$$ we know then the cost of the vertex test $\val_{vertex}(f)$ is $O(1/m^{39})$.
Let $g_{v,v'} = (g_{v,v'}^1,g_{v,v'}^2,\ldots,g_{v,v'}^q)$; $g_{v,v'}\in \Delta_q$   as $g_{v,v'}^i(x) = \E_{u\sim v }[f_{u,v'}^i(\pi_{u,v}(x)]$ and each $(f_{u,v'}^1,f_{u,v'}^2,\ldots,f_{u,v'}^q)\in \Delta_q$ for any $u$, by the definition of $f_{u,v'}^i$ as the indicator function of $f_{u,v'} = i$. For every $g_{v,v'}$, we
classify it into the following three categories:\\
(1) dictator function:   there exists some $g_{v,v'}^i $  with  its
$\delta$-noisy-influence influence above $\tau$ (with $\tau$ being
specified later),\\
(2) constant function: there exists some $i$ with
$\E[g_{v,v'}^i]\geq 1-\delta^{0.1}$, and \\
(3) all the other $g_{v,v'}$ not in category (1) and (2).

The main idea of the remaining proof is to show that for every
$v'\in V'$, in order to  bound  $\val_{vertex}(f)$ by
$O(1/m^{39})$, we must have $g_{v,v'}$ in category $(2)$
for at least a $1-O(1/m^{10})$ fraction of $v \in V$.
Then we argue that this will incur a big cost on the edge test.
We proceed as follows:
\begin{enumerate}[$(i)$]
\item Bound the fraction of $g_{v,v'}$ in category (3): Suppose for a fixed $v' \in V'$ and an $\alpha$-fraction of $v \in V$, $g_{v,v'}$ is in category (3); i.e., $$\max_{i\in R} \Inf_{i}^{1-\delta} [g_{v,v'}] \leq \tau$$
 and  $$\max_{i\in [k]} \E_{x\in [k]^R}[g_{v,v'}^i(x)] < 1-\delta^{0.1}.$$
Then setting $\tau$ to make  $err(\tau,q, \delta)\leq 1/m^{30}$ in Theorem~\ref{thm:MIST}, we have that
\[
\Stab_{1-\delta}[g_{v,v'}^i] \leq \Gamma_{1-\delta}(\mu^i_{v,v'}) + err(\tau,q,\delta).
\]
where $\mu_{v,v'}^i = \E_{x}[g^i_{v,v'}(x)]$. We know that $$\sum_{i=1}^q \mu_{v,v'}^i =\sum_i^q \E_{x}[g^i_{v,v'}(x)] =\sum_{i=1}^q  \E_{x, u\sim v }[f_{u,v'}^i(\pi_{u,v}(x)] = 1.$$
Suppose that $\mu^{i^*}_{v,v'}$ has the maximum value among $\mu^1_{v,v'},\mu^2_{v,v'},\ldots,\mu^k_{v,v'}$, we know then $$\mu^{i^*}_{v,v'}\geq \frac{1}{q}\geq \frac{1}{m^{4}}\geq \frac{1}{\delta^{0.1}}.$$
We apply Lemma~\ref{lem:tech} to $\mu^{i^*}_{v,v'}$, observing that $\delta^{0.1} \leq \mu^{i^*}_{v,v'} \leq 1-\delta^{0.1}$.
We obtain
$ \mu^{i^*}_{v,v'} - \Gamma_{1-\delta}(\mu^{i^*}_{v,v'}) = \Omega(\delta^{0.7}).$
For $i \neq i^*$ we simply use $ \mu^{i}_{v,v'} - \Stab_{1-\delta}(g^{i}_{v,v'}) \geq 0$.
Therefore,
\begin{eqnarray*}
        1-\sum_{i=1}^q\Stab_{1-\delta}[g_{v,v'}^i] & = & \sum_{i=1}^{k} \left( \mu^i_{v,v'} - \Stab_{1-\delta}[g^i_{v,v'}] \right) \\
        &\geq&  \mu^{i^*}_{v,v'}-\Gamma_{1-\delta}(\mu^{i^*}_{v,v'}) - err(\tau,q,\delta) \\
        &\geq& \Omega(\delta^{0.7})- O\left({1 \over m^{30}}\right)\\
        &=&  \Omega\left({1 \over m^{28}} \right).
\end{eqnarray*}

Overall, each particular $v' \in V$ is picked with probability at
least $1/m$, and if $g_{v,v'}$ for an $\alpha$-fraction of $v \in V$ is in category (3),
the vertex test will have cost $\Omega\left(\alpha \cdot {1 \over m} \cdot {1 \over m^{28}}\right)$. 
In order to keep the cost of the vertex test bounded by $O(1/m^{39})$,  we must have $\alpha = O(1/m^{10})$.

\item Bound the fraction of $g_{v,v'}$ in category (1): For a fixed $v'\in V'$, suppose that for a $\beta$-fraction of $v \in V$,  $g_{v,v'}$  has a coordinate with $\delta$-noisy influence above $\tau$. Then consider the following labeling for the \uniquegames instance: for each $v\in V$, we can just assign randomly  a label from the following list:
        $$\Lambda_v = \{i \in R: \Inf^{1-\delta}_i [g_{v,v'}] \geq \tau \}
        \cup \{i\in R: \Inf^{1-\delta}_i [f_{v,v'}] \geq \tau/2\}.$$
Then by Lemma~\ref{lem:ti}, 
\[
\sum_{i=1}^q \Inf^{1-\delta}_i [f_{v,v'}]\leq 1/\delta,
\]
and 
\[
\sum_{i=1}^q \Inf^{1-\delta}_i [g_{v,v'}]\leq 1/\delta,
\]
we know   $|\Lambda_v|= O(\frac{1}{\delta\tau})$.

By Fact~\ref{prop:convexity-noisy-influences}, for every  $i \in \Lambda_v$,
$\E_{u\sim v}[\Inf_i^{1-\delta} [f_{u,v'}(\pi_{v,u}(x))]] \geq \Inf_i ^{1-\delta} [g_{v,v'}] \geq \tau$.
By an averaging argument, at least a $\tau/2$-fraction of neighbors $u \sim v$ have a
coordinate $j \in R$ such that $\Inf_j^{1-\delta} [f_{u,v'}] \geq \tau/2$ and
$\pi_{v,u}(i) = j$. Therefore, at least a $\tau/2$-fraction of edges in the \uniquegames instance
have candidate labels in their lists that satisfy the edge.
Hence choosing the label of $v$ independently and uniformly from $\Lambda_v$ will satisfy
each edge with probability $\Omega(\beta \tau^3 \delta^2)$.
In expectation, we satisfy an $\Omega(\beta \tau^3\delta^2)$-fraction of the edges of the
\uniquegames instance. Since we can take $\eps$ (in the Unique Games Conjecture) to be an arbitrarily small constant.
If we set $\epsilon = \min (\tau^3 \cdot 1/m^4\cdot  \delta^2,1/m^{80})$,
we conclude that $\beta = O(1/m^4)$.
\end{enumerate}

\medskip\noindent
Therefore, we can assume that for every $v' \in V'$,  $g_{v,v'}$ for a $(1-O(1/m^4))$-fraction of $v \in V$ is in category $(2)$; i.e., there exists some $i \in [q]$ such that  $\E[g^i_{v,v'}] =\E_{u\sim v}[f^i_{u,v'}]\geq 1-\delta^{0.1}= 1-1/m^4$. Therefore, if we pick a random $v\in V$, we have that 
$\E_{v\in V}[g^i_{v,v'}] \geq 1- O(1/m^4)$.

Notice the regularity of the graph, we have that 
$$\E_{v\in V}[g^i_{v,v'}]=\E_{v\in V}[\E_{u\sim v}[f^i_{u,v'}]] = \E_{u\in V}[f^i_{u,v'}]  \geq 1-O(1/m^4).$$

By an average argument, for every $v' \in V'$ and for at least a $(1-O(1/m^2))$-fraction of the vertices $u \in V$, we have that $\E[f^{i}_{u,v'}]\geq 1-1/m^2$.
Since $|V'| \leq m$, by a union bound,  we have that for a $(1-O(1/m))$ fraction of the $v\in V$, $\max_i \E[f^i_{v,v'}] \geq 1-1/m$  for every $v'\in V'$. Let us call these $v \in V$ ``good".

Given $(v,v')$ fixed, let us just consider the labeling of $(v,v',x)$ by $\arg \max_{i} \E[f_{v,v'}^i]$ for every $x$. This labeling has a cost at least $s$ as it assigns a  label depending only on $(v,v')$ which can be viewed as the cost of an integral labeling for the gap instance $\calI'$.

Given a good $v$, $f$ \((v_1,v'_{1},\pi_{v,v_1}(x^1)),(v,v_2,\pi_{v_2,v}(x^2)),\ldots,(v_j,v'_j,\pi_{v,v_j}(x^j))\) in the reduction has the same cost as labeling each $(v_i, v_i', \pi_{v_i,v}(x))$ with $\arg \max_{i} \E[f_{v,v'}^i]$
 on $(1-O(j/m))$ fraction of the hyperedges.
Therefore, we have that $\val_{vertex}(f)\geq  (1-O(j/m))\cdot s\geq (1-O(k/m)) \cdot s$. This implies that $\val(f) \geq \eta \val_{vertex}(f)\geq \eta (1-k/m)\cdot s =  s/m^{39}-O(k/m^{40})$ which leads to a contradiction.

It remains to prove the same result holds for Min-$\ss$-LCSP and Min-$\ss$-TCSP. For Min-$\ss$-LCSP. Assuming that there is an integrality gap instance $\calI'$ with $LP(\calI) = c, \opt(\calI')=s$.
For any $v'\in V'$ with the candidate label set $L_{v'}$  we will remove the requirement that its label must be in $L_v'$, instead we will add a unary cost  function 
$P_{L_{v'}}(z):[q]\to \{0,\infty\}$ defined as
$$
P_{L_{v'}}(z)=
\left\{
	\begin{array}{ll}
		0  & \mbox{if } z \in L_{v'} \\
		\infty & \mbox{if } z \notin L_{v'}
	\end{array}
\right.
$$  on $v'$.  Such a change transforms the Min-$\ss$-LCSP instance $\calI'$ into an  Min-$\ss\cup \{P_{L_{v'} }| v'\in V"\}$-CSP instance $\calI''$. 

Since adding the $P_{L_{v'}}$ constraint on $v'$ is essentially the same as restricting the labeling of $v'$ in $L_v'$, we know that $\opt(\calI') = \opt(\calI'')$ and $LP(\calI') = LP(\calI'')$. 
We will then use the reduction from a \uniquegames instance $\calU$ (with the integrality gap instance $\calI''$) to a Min-$\ss\cup \{P_{L_{v'} }| v'\in V"\}$-CSP instance $\calM$,  we still have  that

\begin{enumerate}
\item Completeness property: if \(\opt(\calU)\geq 1-\eps\), then
$\opt(\calM)\leq  \eta(c+O(1/m))$;
\item Soundness property: if $\opt(\calU)\leq \eps$, then
\(\opt(\calM)\geq \eta(s -O(k/m))$.
\end{enumerate} 

Then we can convert   $\calM$ into a Min-$\ss$-LCSP instance $\calM'$. For any vertex $u$ in $\calM$ with  $P_{L_{v'}}$ on it, we will remove $P_{L_{v'}}$ and set the candidate list of $u$ to be $L_{v'}$. By doing this, we removing all the appearance of $P_{L_{v'}}$.  Such a conversion also has the property that $\opt(\calM)=\opt(\calM')$. Therefore, we prove the same \uniquegames hardness result holds for the Min-$\ss$-LCSP.
 
 As for Min-$\ss$-TCSP,   for any terminal vertex $v'\in V$ with fixed label $i$, we will remove the the requirement of the label of $v'$ to be $i$, instead we apply  an equivalent cost function 
$$P_i(z) =\left\{\begin{array}{ll}
		0  & \mbox{if } z =i \\
		\infty & \mbox{if } z \ne i
	\end{array}
	\right.
$$
A similar argument as the proof for Min-$\ss$-LCSP  also holds.
\end{proof}

\section{Lov\'asz versus Basic LP}
\label{app:Lovasz-LD}

Here we compare the Basic LP and the Lov\'asz convex relaxation, for problems that can be phrased both as a Min-CSP and as a Submodular Multiway Partition problem. That is, we consider a \ss-Min-CSP where each predicate $\Psi_e \in \ss$ is of the form
$$ \Psi_e(i_1,\ldots,i_l) = \sum_{i=1}^{k} f_e(\{ j: i_j = i\}) $$
and each $f_e:2^{[l]} \rightarrow \RR_+$ is a submodular function. (Recall that \mc, \hmc and \nwmc fall in this category.) Then in both relaxations, we replace the labeling $x_v \in [k]$ by variables $y_{v,i} \geq 0, i \in [k]$ such that $\sum_{i=1}^{k} y_{v,i} = 1$. We also impose the list-coloring constraints $y_{v,i} = 0$ for $i \notin L_v$ in both cases.
The objective function is defined in different ways in the two relaxations.

\paragraph{The Lov\'asz relaxation.}
We minimize $\sum_{e \in E} \sum_{i=1}^{k} \hat{f_e}(y_{v,i}: v \in e)$, where $\hat{f}$ is the Lov\'asz extension of $f$, $\hat{f}(\by) = \E_{\theta \in [0,1]}[f(A_\by(\theta))]$ where $A_\by(\theta) = \{i: y_i > \theta\}$.

\paragraph{The Basic LP.}
We minimize $\sum_{e \in E} \sum_{i_1,\ldots,i_l \in [k]} y_{e,i_1,\ldots,i_k} \Psi_e(i_1,\ldots,i_l)$
subject to the consistency constraints $\sum_{i_1,\ldots,i_{j-1},i_{j+1},\ldots,i_l} y_{e,i_1,\ldots,i_l} = y_{v,i_j}$
where $v$ is the $j$-th vertex of $e$.

\begin{lemma}
\label{lem:stronger}
The value of the Lov\'asz relaxation is at most the value of the Basic LP.
\end{lemma}

\begin{proof}
Given a fractional solution of the Basic LP, with variables $y_{v,i}$ for vertices $v \in V$ and $y_{e,i_1,\ldots,i_l}$ for hyperedges $e \in E$, each hyperedge contributes $\sum_{i_1,\ldots,i_l \in [k]} y_{e,i_1,\ldots,i_k} \Psi_e(i_1,\ldots,i_l)$, where
$ \Psi_e(i_1,\ldots,i_l) = \sum_{i=1}^{k} f_e(\{ j: i_j = i\})$. In other words, each assignment $(i_1,\ldots,i_l)$ contributes $\sum_{i=1}^{k} f(S_i)$ where $S_i$ is the set of coordinates labeled by $i$. Aggregating all the contributions of a given set $S$, we can define $z_{e,i,S}$ as the sum of $y_{e,i_1,\ldots,i_l}$ over all choices where the coordinates labeled by $i$ are exactly $S$. Then, the contribution of hyperedge $e$ becomes $\sum_{i=1}^{k} \sum_{S \subseteq [l]} z_{e,i,S} f(S)$. Moreover, the variables $z_{e,i,S}$ are consistent with $y_{v,i}$ for $v \in e$ in the sense that $y_{v,i} = \sum_{S: v \in S} z_{e,i,S}$.  

On the other hand, the contribution of a hyperedge $e$ in the Lov\'asz relaxation is given by the Lov\'asz extension $\sum_{i=1}^{k} \hat{f_e}(y_{v,i}: v \in e)$. It is known that the Lov\'asz extension is the convex closure of a submodular function, i.e.~$\hat{f_e}(y_{v,i}: v \in e)$ is the minimum possible value of $\sum_{S \subseteq [l]} z_{e,i,S} f(S)$ subject to the consistency constraints $y_{v,i} = \sum_{S: v \in S} z_{e,i,S}$, and $z_{e,i,S} \geq 0$. Therefore, the contribution of $e$ to the Lov\'asz relaxation is always at most the contribution in the Basic LP.  
\end{proof}

This means that the Basic LP is potentially a tighter relaxation%, since its optimum is closer to the integer optimum,
and its fractional solution carries potentially more information than a fractional solution of the Lov\'asz relaxation.
However, in some cases the two LPs coincide: We remark that for the \mc problem, both relaxations are known to be equivalent to the CKR relaxation (see \cite{CE11a} for a discussion of the Lov\'asz relaxation, and \cite{MNRS08} for a discussion of the ``earthmover LP", identical to our Basic LP). Our results also imply that for the \hmc problem, both relaxations have the same integrality gap, $2-2/k$.
However, for certain problems the Basic LP can be strictly more powerful than the Lov\'asz relaxation.

\paragraph{Hypergraph Multiway Partition.}
{\em Given a hypergraph $H = (V,E)$ and $k$ terminals $t_1,\ldots,t_k \in V$, find a partition $(A_1,\ldots,A_k)$ of the vertices,
so that $\sum_{i=1}^{k} f(A_i)$ is minimized, where $f(A)$ is the number of hyperedges cut by $(A,\bar{A})$.}

\medskip
\noindent
This is a special case of \submpsym, because $f$ here is the cut function in a hypergraph, a symmetric submodular function.
The difference from Hypergraph Multiway Cut is that in Hypergraph Multiway Partition, each cut hyperedge contributes $1$ for each terminal that gets assigned some of its vertices (unlike in \hmc where such a hyperedge contributes only $1$ overall).

\begin{lemma}
There is an instance of Hypergraph Multiway Partition where the Lov\'asz relaxation has a strictly lower optimum than the Basic LP.
\end{lemma}

\begin{proof}
The instance is the following: We have a ground set $V = \{ t_1, t_2, t_3, t_4, t_5, a_{12}, a_{23}, a_{34}, a_{45}, a_{51} \}$. The hyperedges are $\{ t_1, t_2, a_{12} \}$, $\{ t_2, t_3, a_{23} \}$, $\{ t_3, t_4, a_{34} \}$, $\{ t_4, t_5, a_{45} \}$, $\{ t_5, t_1, a_{51} \}$ with unit weight, and $\{ a_{12}, a_{23}, a_{34}, a_{45}, a_{51} \}$ with weight $\epsilon = 0.001$. 

The idea is that fractionally, each non-terminal $a_{i,i+1}$ must be assigned half to $t_i$ and half to $t_{i+1}$, otherwise the cost of cutting the triple-edges is prohibitive. Then, the cost of the 5-edge $\{ a_{12}, a_{23}, a_{34}, a_{45}, a_{51} \}$ is strictly higher in the Basic LP, since there is no good distribution consistent with the vertex variables (while the Lov\'asz relaxation is not sensitive to this). We verified by an LP solver that the two LPs have indeed distinct optimal values.
\end{proof}

\section{The equivalence of integrality gap and symmetry gap}
\label{app:integr-sym}

In this section, we prove that for any Min-\ss-CSP (specified by allowed predicates and candidate lists), the worst-case integrality gap of its Basic LP and the worst-case symmetry gap of its multilinear relaxation are the same. As we have seen, if we consider a \ss-Min-CSP problem and $\ss$ includes the Not-Equal predicate, then the integrality gap of the Basic LP implies a matching inapproximability result assuming UGC. Similarly, if the objective function can be viewed as $\sum_{i=1}^{k} f(S_i)$ where $f$ is submodular and $S_i$ is the set of vertices labeled $i$, then the symmetry gap implies an inapproximability result for this ``submodular generalization" of the \ss-CSP problem. In fact, the two hardness threshold often coincide as we have seen in the case of \hmc and \submp, where they are equal to $2-2/k$. 
Here, we show that it is not a coincidence that the two hardness thresholds are the same. 

We recall that the symmetry gap of a Min-CSP is the ratio $s/c$ between the optimal fractional solution $c$ and the optimal symmetric fractional solution $s$ of the {\em multilinear relaxation}. The objective function in the multilinear relaxation is $F(\bx) = \E[f(\hat{\bx})]$,
where $\hat{\bx}$ is an integer solution obtained by independently labeling each vertex $v$ by $i$ with probability $x_{v,i}$. The notion of symmetry here is that  $\calI'$ on a ground set $V'$ is invariant under a group $\cal{G}$ of permutations of $V'$. A fractional solution is symmetric if for any permutation $\sigma \in \cal{G}$ and $v'\in V'$, $v'$ and $\sigma(v')$ have the same fractional assignment.

The following theorem states that an integrality gap instance can be converted into a symmetry gap instance.

\begin{theorem} For any $\ss$-CSP instance $\calI(V,E,k,L_v,h)$ whose Basic LP has optimum $LP(\calI) = c$,
and the integer optimum is $\opt(\calI) = s$, there is a symmetric $\ss$-CSP instance $\calI'(V',E',k,L_v',h)$ whose symmetry gap is $s/c$. 
\end{theorem}

\begin{proof}
Given an optimal solution $\bx$ of the Basic LP for instance $\calI$, without loss of generality, let us assume that all the variables in the solution have a rational value and there exists some $M$ such that the values of all the variables, i.e, $x_{v,i}$ and $x_{e,\alpha}$, in the LP solution become integers if multiplied by $M$. 
For every vertex $v\in V$, we define
	$$S_v = \{y \in [k]^M : \text{for every $i\in [k]$;} \;\;
	\text{$x_{v,i}$ fraction of $y$'s coordinates have value
	$i$}\}.$$
In other words, $S_v$ is the collection of strings in $[k]^M$ such that the portion of appearances of $i$ in each string is exactly $x_{v,i}$.

We define a new instance on a ground set $V' = \{ (v,y): v \in V, y \in S_v \}$.  For every vertex $(v,y)$ in $V'$, its candidate list is $L_v$,
the candidate list of vertex $v$ in instance $\cI$. Given an edge $e=(v_1,v_2,\ldots,v_{|e|})\in E$ and its local distribution $\calP_e$ over $[k]^{|e|}$ (implied by the fractional solution $x_{e,\alpha}$), we call $(y^1,y^2,\ldots, y^{|e|}) \in ([k]^M)^{|e|}$ ``consistent with $e$" if for every $\alpha\in [k]^{|e|}$, the fraction of $i \in [M]$ such that $(y^1_i,y^2_i,\ldots,y^{|e|}_i)=\alpha$ is exactly $x_{e,\alpha}$. It is easy to check (using the LP consistency constraints) that each $y^i$ must be  in $S_{v_i}$.
We define $S_e$ to be the collection of all possible $(y^1,y^2,\ldots,y^{|e|})$ that are consistent with edge $e$.

For every edge $e=(v_1,v_2,\ldots,v_{|e|})$ and  every  $(y^1,y^2,\ldots,y^{|e|})$ consistent with $e$, we add a constraint $\Psi_e$ over
$(v_1,{y^1}),\ldots,(v_j,{y^{|e|}})$ with the same predicate function as edge $e$. 
We assign equal weights to all these copies of constraint $\Psi_e$, so that their total weight is equal to $w_e$.  

Let $\cal{G}$ be the group of all permutations $\pi:[M]\to [M]$. For any $\pi\in \cal{G}$, we also use the notation $\pi(y)$ to indicate $(y_{\pi(1)},y_{\pi(2)},\ldots,y_{\pi(M)})$ for any $y\in [k]^M$. Let us also think of $\cal{G}$ as a permutation on $V'$: for any $(v,y)\in V'$, we map it to $(v,\pi(y))$.  
Then  it is easy to check that $\cal I'$ is invariant under any
permutation $\sigma\in \cal{G}$. Also, for any $(v,y), (v,y') \in
V'$, since $y$ and $y'$ contain the same number of occurences for
each $i \in [M]$ (consistent with $x_{v,i}$), we know there exists
some $\pi \in \cal{G}$ such that $\pi(y) = y'$. Therefore, a
fractional solution of $\calI'$ is symmetric with respect to $\cG$ if
and only if $(v,y)$ has the same fractional assignment as $(v,y')$, for all $y,y' \in S_v$.

Let us also assume that both $\calI$ and $\calI'$ are normalized with
all their weights summing to $1$; i.e, there is a distribution over the constraints with the probability being the weights. Essentially,  $\calI'$  is constructed by randomly picking a constraint $\Psi_e(v_1,v_2,\ldots,v_{|e|})$ in $\calI$  and then randomly picking  $(y^1,y^2,\ldots,y^{|e|})$ in $S_e$. Then we add a constraint $\Psi_e$ on $(v_1,y^1),(v_2,y^2),\ldots,(v_{|e|},y^{|e|})$ in $\calI'$.

Below, we prove that assuming there is a gap of $s$ versus $c$ between integer and fractional solutions for the Basic LP of $\calI$, there is a gap of at least $s$ versus $c$ between the best symmetric fractional solution and the best asymmetric fractional solution for the multilinear relaxation of $\calI'$.

\paragraph{Some (asymmetric) solution of $\calI'$ has cost at most $c$:} Consider the following (random) assignment for $\calI'$: we first choose  a random $i\in [M]$ and then assign to each $(v,y)$ the label $y_i$. We know $y_i \in L_v$, as it is a valid assignment for $v$ by the definition of $S_v$. Such an assignment has expected cost
	$$\E_{i\in[M],e,y^1,\ldots,y^{|e|}}[\Psi_e(y^1_i,\ldots,y^{|e|}_i)] =
	\E_{e}[\sum_{\alpha\in [k]^{|e|}} x_{e,\alpha}\Psi(\alpha)]
	= LP(\calI) = c.$$
Therefore there is also some deterministic assignment of cost at most $c$. Our assignment is in fact integral, so the multilinear relaxation does not play a role here.

\paragraph{Every symmetric fractional solution of $\calI'$ has cost at least $s$:} We know that a symmetric solution gives the same fractional assignment $x_v=(x_{v,1},x_{v,2},\ldots,x_{v,k})$ to vertex $(v,y)$ for every $y\in S_v$. Also by definition of the multilinear relaxation, the value of this fractional solution is equal to the expected cost of independently choosing the label of each $(v,y)$ from the distribution $\calP_v$, which means $i$ with probability $x_{v,i}$. Therefore, given a symmetric solution of $\calI'$ specified by $\{x_v \mid v\in V\}$, it has cost
\begin{equation}\label{eqn:cost}
       \E_{e=(v_1,\ldots,v_{|e|})\in E,(y^1,\ldots,y^{|e|})\in S_e, l_i \sim \calP_{v_i}}[\Psi_e(l_1,l_2,\ldots,l_{|e|})]
\end{equation}
Since $l_1,l_2,\ldots,l_{|e|}$ is independent of $(y^1,y^2,\ldots,y^{|e|})$, we have
\[
\eqref{eqn:cost} =  \E_{e=(v_1,v_2,\ldots,v_{|e|})\in E, l_1,l_2,\ldots,l_{|e|}}[\Psi_e(l_1,l_2,\ldots,l_{|e|})]
\]
where $l_i$ is chosen independently from the distribution $\calP_{v_i}$ for the respective vertex $v_i \in e$.
This is exactly the cost on the original instance $\calI$ if we independently label each $v$ by sampling from distribution $\calP_v$. Since the integer optimum of $\calI$ is at least $s$, we have that $\eqref{eqn:cost}$ is also at least $s$.
\end{proof}

\medskip\noindent
In the other direction, we have the following.

\begin{theorem} For any symmetric $\ss$-CSP instance $\calI(V,E,k,L_v,h)$ whose multilinear relaxation has symmetry gap $\gamma$, and for any $\epsilon>0$ there is a $\ss$-CSP instance $\calI'(V',E',k,L_v',h)$ whose Basic LP has integrality gap at least $(1-\epsilon) \gamma$. 
\end{theorem}

\begin{proof}
Assume that $\cal I$ is an instance symmetric under a group of permutations $\cG$ on the vertices $V$. For each $v \in V$, we denote by $\omega(v)$ the {\em orbit} of $v$, $\omega(v) = \{ \sigma(v): \sigma \in \cG \}$. We produce a new instance $\cal I'$ by {\em folding} the orbits, i.e.~we identify all the elements in a given orbit.

First, let us assume for simplicity that no constraint (edge) in $\cal I$ operates on more than one variable from each orbit. Then, we just fold each orbit $\omega(v)$ into a single vertex. I.e., $V' = \{ \omega(v): v \in V \}$. (We abuse notation here and use $\omega(v)$ also to denote the vertex of $V'$ corresponding to the orbit $\omega(v)$.) We also define $E'$ to be edges corresponding one-to-one to the edges in $E$, with the same predicates; i.e.~each constraint $\Psi_e(x_1,\ldots,x_{|e|})$ in $\cal I$ becomes $\Psi_e(\omega(x_1), \ldots, \omega(x_{|e|}))$ in $\cal I'$. The candidate list for $\omega(v)$ is identical to the candidate list of $v$ (and hence also to the candidate list for any other $w \in \omega(v)$, by the symmetry condition on $\cal I$). 

Now consider an optimal fractional solution of the multilinear relaxation of $\cal I$, with variables $x_{v,i}$. We define a fractional solution of the Basic LP for $\cal I'$, by setting $x_{\omega(v),i} = \frac{1}{|\omega(v)|} \sum_{w \in \omega(v)} x_{w,i}$. We claim that the edge variables $x_{e,\alpha}$ can be assigned values consistent with $x_{\omega(v),i}$ in such a way that the objective value of the Basic LP $\cal I'$ is equal to the value of the multilinear relaxation for $\cal I$. This is because the value of the multilinear relaxation is obtained by independently labeling each vertex $v$ by $i$ with probability $x_{v,i}$. We can then define $x_{e,\alpha}$ as the probability that edge $e$ is labeled $\alpha$ under this random labeling. This shows that the optimum of the Basic LP for $\cal I'$ is at most the optimum of the multilinear relaxation of $\cal I$ (if we have a minimization problem; otherwise all inequalities are reversed).

Now, consider an integer solution of the instance $\cal I'$, i.e. a labeling of the orbits $\omega(v)$ by labels in $[k]$. This induces a natural symmetric labeling of the instance $\cal I$, where all the vertices in each orbit receive the same label. The values of the two labelings of $\cal I$ and $\cal I'$ are the same. (Both labelings are integer, so the multilinear relaxation does not play a role here.) This shows that the symmetric optimum of the multilinear relaxation of $\cal I$ is at most the integer optimum of $\cal I'$. Together with the previous paragraph, we obtain that the gap between integer and fractional solutions of the Basic LP of $\cal I'$ is at least the gap between symmetric and asymmetric solutions of the multilinear relaxation of $\cal I$.
This completes the proof in the special case where no predicate takes two variables from the same orbit.

\medskip

Now, let us consider the general case, in which a constraint of $\cal I$ can contain multiple variables from the same orbit. Note that the proof above doesn't work here, because it would produce predicates in $\cal I'$ that operate on the same variable multiple times, which we do not allow. (As an instructive example, consider a symmetric instance on two elements $\{1,2\}$ with the constraint $l_1 \neq l_2$, which serves as a starting point in proving the hardness of $(1/2+\epsilon)$-approximating the maximum of a nonnegative submodular function. This instance is symmetric with respect to switching $1$ and $2$, and hence the folded instance would contain only one element and the constraint $l'_1 \neq l'_1$, which is not a meaningful integrality gap instance.)

Instead, we replace each orbit by a large cluster of identical elements, and we produce copies of each constraint of $\cal I$, using distinct elements from the respective clusters. As an example, consider the 2-element instance above. We would first fold the elements $\{1,2\}$ into one, and then replace it by a large cluster $C$ of identical elements. The original constraint $l_1 \neq l_2$ will be replaced by the same constraint for every pair of distinct elements $\{i,j\} \subset C$. In other words, the new instance is a Max Cut instance on a complete graph and we consider the Basic LP for this instance. The symmetry gap of the original instance is $2$, and the integrality gap of the Basic LP is arbitrarily close to $2$ (for $|C| \rightarrow \infty$).

Now let us describe the general reduction. For each orbit $\omega(v)$ of $\cal I$, we produce a disjoint cluster of elements $C_{\omega(v)}$. The candidate list for each vertex in $C_{\omega(v)}$ is the same as the candidate list for any vertex in $\omega(v)$ (which must be the same by the symmetry assumption for $\cal I$). The ground set of $\cal I'$ is $V' = \bigcup_{v \in V} C_{\omega(v)}$. For each edge  $e = (v_1,\ldots,v_{|e|})$ of $\cal I$, we produce a number of copies by considering all possible edges $e' = (v'_1,\ldots,v'_{|e|})$ where $v'_i \in C_{\omega(v_i)}$ and $v'_1,\ldots,v'_{|e|}$ are distinct. We use the same predicate, $\Psi_{e'} = \Psi_e$. The edge weights $w_{e'}$ are defined so that they are equal and add up to $w_e$.
Let us assume that ${\cal I}, {\cal I'}$ are minimization problems. Let $c$ be value of the optimal solution of the multilinear relaxation of $\cal I$, and let $s$ be the the value of the optimal symmetric solution of the multilinear relaxation of $\cal I$.

\medskip
{\bf The fractional optimum of $\cal I'$ is at most $c$:}
Let $x_{v,i}$ be an solution of the multilinear relaxation of value $c$. We define a fractional solution of the Basic LP for $\cal I'$ to be equal to $x_{v',i} = \frac{1}{|\omega(v)|} \sum_{w \in \omega(v)} x_{w,i}$ for each $v' \in C_{\omega(v)}$. Similar to the construction in the simpler case above, we define the edge variables $x_{e',\alpha}$ to simulate the independent rounding that defines the value of the multilinear extension of $\cal I$; specifically, we define $x_{e',\alpha}$ for $e' = (\omega(v_1), \omega(v_2), \ldots, \omega(v_{|e|}))$ to be the probability that $e = (v_1,\ldots,v_{|e|})$ receives assignment $\alpha$, where $e$ is the edge that produced $e'$ in our reduction. By construction, the value of this fractional solution for $\cal I'$ is equal to the value of the multilinear relaxation for $\cal I$ which is $c$.

\medskip
{\bf The integer optimum of $\cal I'$ is at least $s$:}
Here, we consider any integer solution of $\cal I'$ and we produce a symmetric fractional solution for $\cal I'$.
If the integer solution for $\cal I'$ is denoted by $l(v')$, we consider each cluster $C_{\omega(v)}$ in $\cal I'$ and we define $x_{v,i}$ to be the fraction of vertices in $C_{\omega(v)}$ that are labeled $i$. Note that by definition, $x_{v,i} = x_{w,i}$ for all $w \in \omega(v)$ and hence $x_{v,i}$ depends only on the orbit of $v$. This means that $x_{v,i}$ is a symmetric fractional solution. Also, it respects the candidate lists of $\cal I$ because $\cal I'$ has the same candidate lists and hence the fractional solution uses only the allowed labels for each vertex $v$. The value of this symmetric fractional solution is given by the multilinear extension, i.e.~the expected value of a random labeling where each label is chosen independently with probabilities $x_{v,i}$. The final observation is that in the limit as $|C_{\omega(v)}| \rightarrow \infty$, this is very close to the value of the integer solution $l(v')$ for $\cal I'$: this is because the edges $e'$ in $\cal I'$ are defined by all possible choices of distinct elements in the respective clusters. In the limit, this is arbitrarily close to sampling independently random vertices from the respective clusters, which is what the multilinear relaxation corresponds to. Therefore, the optimal symmetric fractional solution to $\cal I$ cannot be worse than the integer optimum of $\cal I'$ in the limit.

\medskip
In conclusion, the integrality gap of the Basic LP for $\cal I'$ can be made arbitrarily close to the symmetry gap of $\cal I$, $s/c$.
\end{proof}

\bibliographystyle{plain}
\bibliography{multiway}

\appendix

% !TEX root = multiway-main.tex
\section{Missing proofs from Section~\ref{sec:oracle-hardness}}
\label{app:oracle-hardness}

\subsection{Hardness for symmetric $\submp$}

Here, we prove Theorem~\ref{thm:sym-SMP-oracle-hardness}. We follow the symmetry gap approach
of the proof of Theorem~\ref{thm:SMP-hardness}, only the symmetric instance is different here.

\paragraph{The symmetric instance.}
Let $V = [k] \times [k]$. It is convenient to think of the elements
of $V$ as belonging to ``rows'' $R_i = \{(i, j) \;|\; 1 \leq j \leq
k\}$ and ``columns'' $C_i = \{(j, i) \;|\; 1 \leq j \leq k\}$. The
elements $(i, i)$ are the terminals, i.e, $t_i = (i, i)$. We define
our cost function $f$ as follows.

Let $\gamma \in [0,k]$ be a parameter; we assume for convenience that $\gamma$ is an even integer.
Let $\phi: \RR \rightarrow \RR$ be the following function:
\begin{equation*}
        \phi(t) =
        \begin{cases}
                t \qquad\qquad\qquad \text{ if $t \leq k - \gamma/2$},\\
                2k - t - \gamma \qquad \text{ otherwise}.
        \end{cases}
\end{equation*}
For each $i \in [k]$, let $g_i: 2^V \rightarrow \mathbb{R}_+$ be the
following function:
\begin{equation*}
        g_i(S) =
        \begin{cases}
                \phi(|S|) \qquad\qquad \text{ if $t_i \notin S$},\\
                \phi(k - |S|) \qquad \text{ otherwise}.
        \end{cases}
\end{equation*}
The function $f: 2^V \rightarrow \mathbb{R}_+$ is defined by 
$f(S) = \sum_{i = 1}^k g_i(S \cap R_i) + \sum_{i = 1}^k g_i(S \cap C_i)$. It
is straightforward to verify that $f$ is non-negative, submodular,
and symmetric. As before, the instance is symmetric under swapping the vertices $(i,j)$ and $(j,i)$ for all $i \neq j$.

\paragraph{Computing the symmetry gap.}
\begin{theorem}
\label{thm:sym-smp-gap}
        Let $\gamma = 2 \lfloor 2k - \sqrt{3k^2 - 2k} \rfloor$. Then the gap between the best
        symmetric and asymmetric assignment for the instance defined
        above is at least ${8\sqrt{3} - 12 \over 2\sqrt{3} - 2} - O({1 \over k})
		\approx 1.2679 - O({1 \over k})$.
\end{theorem}

\begin{definition}
\label{def:nice}
        We say that element $(i, j)$ is \textbf{unhappy} if it is
        assigned to a terminal $\ell \notin \{i, j\}$.  An assignment is
        \textbf{nice} if it has the following properties:
        \begin{enumerate}[(a)]
        \item Each unhappy element is assigned to terminal $1$.
        \item For each $i \neq 1$, the number of elements in row $R_i$
        that are assigned to terminal $i$ is either $0$ or at least
        $\gamma / 2$.  Similarly, for each $i \neq 1$, the number of
        elements in column $C_i$ that are assigned to terminal $i$ is
        either $0$ or at least $\gamma / 2$. The number of elements in
        row $R_1$ and column $C_1$ that are assigned to terminal $1$ is
        at least $\gamma / 2$.
        \item For each $i \neq 1$, if the number of unhappy elements in
        row $R_i$ is at least $k - \gamma/ 2$, all the elements in $R_i \setminus 
        \{(i, i)\}$ are assigned to terminal $1$. Similarly, for each $i
        \neq 1$, if the number of unhappy elements in column $C_i$ is at
        least $k - \gamma/ 2$, all the elements in $C_i \setminus  \{(i, i)\}$ are
        assigned to terminal $1$.
        \end{enumerate}
\end{definition}

\noindent
Theorem~\ref{thm:sym-smp-gap} will follow from the following lemmas.

\begin{lemma}
\label{lem:nice}
        Consider a symmetric assignment of cost $C$. Then there
        exists a nice symmetric assignment whose cost is at most $C + 6k$.
\end{lemma}

\begin{lemma}
\label{lem:cost-symmetric-nice}
        Let $\gamma = 2 \lfloor 2k - \sqrt{3k^2 - 2k} \rfloor$.  Then the cost of any nice
        symmetric assignment is at least $(8\sqrt{3} - 12)k^2 - O(k)$.
\end{lemma}

\begin{proof}[Proof of Theorem~\ref{thm:sym-smp-gap}]
        Note that it follows from Lemma~\ref{lem:nice} and
        Lemma~\ref{lem:cost-symmetric-nice} that the optimal symmetric
        assignment has cost
                $$\overline{\opt} \geq (8\sqrt{3} - 12)k^2 - O(k).$$
        Now we claim that the optimal asymmetric assignment has cost
                $$\opt \leq 2k^2 - \gamma k \leq 2k^2 - 4k^2 + 2k \sqrt{3k^2} + O(k)
                = (2\sqrt{3} - 2)k^2 + O(k).$$
        Consider the assignment that assigns the $i$-th row $R_i =
        \{(i, j) \;|\; 1 \leq j \leq k\}$ to terminal $i$. We have
        \begin{align*}
                f(R_i) &= \sum_{\ell = 1}^k g_{\ell}(R_i \cap R_{\ell}) +
                \sum_{\ell = 1}^k g_{\ell}(R_i \cap C_{\ell})\\
                &= g_i(R_i) + g_i(\{(i, i)\}) + \sum_{\ell \neq i}
                g_{\ell}(\{(i, \ell)\})\\
                &= \phi(0) + \phi(k - 1) + (k - 1)\phi(1)\\
                &= 0 + (k + 1 - \gamma) + (k - 1) \cdot 1 = 2k - \gamma.
        \end{align*}
        Therefore the total cost of the assignment is $k(2k - \gamma)$
        and the theorem follows by plugging in the value of $\gamma$. 
\end{proof}

\subsubsection{Proof of Lemma~\ref{lem:nice}}
Consider a symmetric assignment. 
Let $A_{\ell}$ denote the set of all elements that are assigned to
terminal $\ell$.  Note that, for each $i$, row $R_i$ and column $C_i$
have the same number of unhappy elements; additionally, they have the
same number of elements assigned to terminal $i$. Let $n_i = |A_i \cap R_i|$
denote the number of elements in $R_i$ that are assigned to terminal
$i$ and let $u_i$ denote the number of elements in $R_i$ that are unhappy.
Let $u_{i\ell}$ denote the number of unhappy elements in $R_i$
that are assigned to terminal $\ell$. Let $I$ be the set of all indices
$i$ such that $n_i < \gamma/2$.

We reassign elements as follows. For each $i \in I \setminus \{1\}$, i.e.
such that row $R_i$ has fewer than $\gamma / 2$ elements assigned to terminal $i$,
we reassign all the elements in $(R_i \cup C_i) \cap (A_i \setminus \{(i,i)\})$
to terminal $1$. If row $R_1$ has less than $\gamma / 2$
elements assigned to terminal $1$, we reassign all the elements in
$R_1 \cup C_1$ to terminal $1$. We reassign each remaining unhappy
element to terminal $1$.  Finally, for each $i$ such that row $R_i$
has at least $k - \gamma/2$ elements assigned to terminal $1$, we
reassign all elements in $R_i \cup C_i \setminus \{(i, i)\}$ to terminal $1$.
Note that the resulting assignment is symmetric.
        
Let $(A'_1, \cdots, A'_{k})$ denote the resulting assignment, where
$A'_{\ell}$ is the set of all elements that are assigned to terminal
$\ell$, and let $C'$ be its cost. Let $J$ denote the set of all indices
$i$ such that all of the elements of $R_i \setminus  \{(i, i)\}$ have been
reassigned to terminal $1$; more precisely, $A'_1 \cap R_i \supseteq
R_i \setminus  \{(i, i)\}$.

\begin{claim}
        For each $i \in I$, $g_i(A'_i \cap R_i) \leq g_i(A_i \cap R_i) -
        n_i + 1$. For each $i \in \bar{I}$, $g_i(A'_i \cap R_i) = g_i(A_i
        \cap R_i)$.
\end{claim}
\begin{proof}
        Consider an index $i \in I$. Suppose first that $i = 1$. Note that we
        reassigned all of the elements in $R_1$ to terminal $1$. Thus
        $g_1(A'_1 \cap R_1) = \phi(k - |A'_1 \cap R_1|) = 0$. Since $n_1
        < \gamma/2$ and $\gamma \leq k$, $g_1(A_1 \cap R_1) = \phi(k - |A_1 \cap R_1|)
         = k -\gamma + n_1 \geq n_1$. Therefore we may assume that $i \neq 1$.
        Note that we reassigned all of the elements in $A_i \cap R_i \setminus 
        \{(i, i)\}$ to terminal $1$ and thus $|A'_i \cap R_i| = 1$.  We
        have $g_i(A_i \cap R_i) = \phi(k - |A_i \cap R_i|) = k - \gamma +
        n_i$ and $g_i(A'_i \cap R_i) = \phi(k - |A'_i \cap R_i|) = k -
        \gamma + 1$, which $g_i(A_i \cap R) \leq g_i(A_i \cap R) - n_i + 1$.

        Consider an index $i \in \bar{I}$. Since all of the elements of
        $A_i \cap R_i$ are happy, we have $A'_i \cap R_i = A_i \cap R_i$
        and therefore $g_i(A'_i \cap R_i) = g_i(A_i \cap R_i)$.
\end{proof}

\begin{corollary}
        We have $\sum_{i = 1}^k g_i(A'_i \cap R_i) \leq \sum_{i = 1}^k
        g_i(A_i \cap R_i) - \sum_{i \in I} n_i + k$.
\end{corollary}

\begin{claim}
        For each $i \in J - \{1\}$, $g_i(A'_1 \cap R_i) \leq \sum_{\ell
        \neq i} g_i(A_{\ell} \cap R_i) - n_i + 1$. For each $i \in I \cap
        (\bar{J} - \{1\})$, $g_i(A'_1 \cap R_i) \leq g_i(A_1 \cap R_i) +
        u_i - u_{i1} + n_i$.  For each $i \in \bar{I} \cap (\bar{J} -
        \{1\})$, $g_i(A'_1 \cap R_i) \leq g_i(A_1 \cap R_i) + u_i -
        u_{i1}$.
\end{claim}
\begin{proof}
        Consider an index $i \in J - \{1\}$.  Since we reassigned all
        elements of $R_i - \{(i, i)\}$ to terminal $1$, $g_i(A'_1 \cap
        R_i) = \phi(k - 1) = k - \gamma + 1$. Suppose there exists an
        $\ell^* \neq i$ such that $|A_{\ell^*} \cap R_i| > k - \gamma/2$.
        Then $g_i(A_{\ell^*} \cap R_i) = \phi(|A_{\ell^*} \cap R_i|) = 2k -
        \gamma - |A_{\ell^*} \cap R_i|$. Additionally, for each $\ell
        \notin \{\ell^*, i\}$, $|A_{\ell} \cap R_i| < \gamma/2$ and
        thus $g_i(A_{\ell} \cap R_i) = \phi(|A_{\ell} \cap R_i|) =
        |A_{\ell} \cap R_i|$. Therefore
        \begin{align*}
                \sum_{\ell \neq i} g_i(A_{\ell} \cap R_i) &= 2k - \gamma -
                |A_{\ell^*} \cap R_i| + \sum_{\ell \notin \{\ell^*, i\}}
                |A_{\ell} \cap R_i|\\
                &= 2k - \gamma - 2|A_{\ell^*} \cap R_i| + k - n_i\\
                &\geq 3k - \gamma - n_i - 2(k - n_i) \qquad\qquad
                \mbox{(since $|A_{\ell^*} \cap R_i| \leq k - n_i$)}\\
                &= k - \gamma + n_i.
        \end{align*}
        Therefore $g_i(A'_1 \cap R_i) = k - \gamma + 1 \leq \sum_{\ell
        \neq i} g_i(A_{\ell} \cap R_i) - n_i + 1$.

        Thus we may assume that $|A_{\ell} \cap R_i| \leq k - \gamma/2$
        for each $\ell \neq i$. For each $\ell \neq i$, we have
        $g_i(A_{\ell} \cap R_i) = \phi(|A_{\ell} \cap R_i|) = |A_{\ell}
        \cap R_i|$. Thus
        \begin{align*}
                \sum_{\ell \neq i} g_i(A_{\ell} \cap R_i) &= k - n_i,\\
                g_i(A'_1 \cap R_i) &= k - \gamma + 1\\
                &= \sum_{\ell \neq i} g_i(A_{\ell} \cap R_i) + n_i - \gamma + 1\\
                &\leq \sum_{\ell \neq i} g_i(A_{\ell} \cap R_i) - n_i + 1
                \qquad\qquad \mbox{(since $\gamma \geq 0$)}.
        \end{align*}

        \noindent
        
        Consider an index $i \in I \cap (\bar{J} \setminus \{1\})$. Then $n_i <
        \gamma/2$ and $u_i + n_i < k - \gamma/2$.  The elements of $R_i$
        that we reassigned to terminal $1$ were the unhappy elements and
        the elements of $A_i \cap R_i$.  Note that, since $i \neq 1$, at
        most one element of $A_1 \cap R_i$ is happy (namely the element
        $(i,1) \in C_1$). Therefore $|A'_1 \cap R_i| \leq u_i + n_i + 1
        \leq k - \gamma/2$ and $|A_1 \cap R_i| \leq u_{i1} + 1 \leq k -
        \gamma/2$. Thus $g_i(A'_1 \cap R_i) = \phi(|A'_1 \cap R_i|) =
        |A'_1 \cap R_i| = |A_1 \cap R_i| + u_i - u_{i1} + n_i$ and
        $g_i(A_1 \cap R_i) = \phi(|A_1 \cap R_i|) = |A_1 \cap R_i|$.

        Consider an index $i \in \bar{I} \cap (\bar{J} \setminus  \{1\})$. Then
        $n_i \geq \gamma/2$ and $u_i < k - \gamma/2$.  The elements of
        $R_i$ that we reassigned to terminal $1$ were the unhappy
        elements.  Note that at most one element of $A_1 \cap R_i$ is
        happy and therefore $|A'_1 \cap R_i| \leq u_i + 1 \leq k -
        \gamma/2$. Additionally, $|A_1 \cap R_i| \leq u_{i1} + 1 \leq k -
        \gamma/2$. Thus $g_i(A'_1 \cap R_i) = \phi(|A'_1 \cap R_i|) =
        |A'_1 \cap R_i| = |A_1 \cap R_i| + u_i - u_{i1}$ and $g_i(A_1
        \cap R_i) = \phi(|A_1 \cap R_i|) = |A_1 \cap R_i|$. 
\end{proof}

\begin{claim}
        For each $i \in J$ and each $\ell \notin \{i, 1\}$,
        $g_i(A'_{\ell} \cap R_i) = 0$. For each $i \in \bar{J}$ and each
        $\ell \notin \{i, 1\}$, $g_i(A'_{\ell} \cap R_i) \leq
        g_i(A_{\ell} \cap R_i) - u_{i\ell}$.
\end{claim}
\begin{proof}
        Consider an index $i \in J$. Then $R_i \setminus  \{(i, i)\} \subseteq
        A'_1$. It follows that, for each $\ell \notin \{i, 1\}$,
        $A'_{\ell} \cap R_i = \emptyset$ and $g_i(A'_{\ell} \cap R_i) =
        \phi(|A'_{\ell} \cap R_i|) = 0$.

        Consider an index $i \in \bar{J}$ and an index $\ell \notin \{i,1\}$.
        Since $i \in \bar{J}$, we have $u_i < k - \gamma/2$. Note
        that at most one element of $A_{\ell} \cap R_i$ is happy and thus
        $|A_{\ell} \cap R_i| \leq u_{i\ell} + 1 \leq k - \gamma/2$.
        Moreover, $|A'_{\ell} \cap R_i| = |A_{\ell} \cap R_i| - u_{i\ell}
        \leq 1$. Thus $g_i(A_{\ell} \cap R_i) = \phi(|A_{\ell} \cap R_i|)
        = |A_{\ell} \cap R_i|$ and $g_i(A'_{\ell} \cap R_i) =
        \phi(|A'_{\ell} \cap R_i|) = |A'_{\ell} \cap R_i| = |A_{\ell} \cap R_i| - u_{i\ell}$. 
\end{proof}

\begin{corollary}
        We have
                $$\sum_{i \in J - \{1\}} \sum_{\ell \neq i} g_i(A'_{\ell}
                \cap R_i) \leq \sum_{i \in J - \{1\}} \sum_{\ell \neq i}
                g_i(A_{\ell} \cap R_i) - \sum_{i \in J - \{1\}} n_i + k.$$
\end{corollary}
\begin{proof}
        By the preceding claims, we have
        \begin{align*}
                \sum_{i \in J - \{1\}} g_i(A'_1 \cap R_i) &\leq \sum_{i \in J
                - \{1\}} \sum_{\ell \neq i} g_i(A_{\ell} \cap R_i) - \sum_{i
                \in J - \{1\}} n_i + k\\
                \sum_{i \in J - \{1\}} \sum_{\ell \notin \{i, 1\}}
                g_i(A'_{\ell} \cap R_i) &= 0\\
                \sum_{i \in J - \{1\}} \sum_{\ell \neq i} g_i(A'_{\ell} \cap
                R_i) &= \sum_{i \in J - \{1\}} g_i(A'_1 \cap R_i) + \sum_{i
                \in J - \{1\}} \sum_{\ell \notin \{i, 1\}} g_i(A'_{\ell} \cap
                R_i)\\
                &\leq \sum_{i \in J - \{1\}} \sum_{\ell \neq i} g_i(A_{\ell}
                \cap R_i) - \sum_{i \in J - \{1\}} n_i + k.
        \end{align*}
\end{proof}

\begin{corollary}
        We have
                $$\sum_{i \in \bar{J} - \{1\}} \sum_{\ell \neq i}
                g_i(A'_{\ell} \cap R_i) \leq \sum_{i \in \bar{J} - \{1\}}
                \sum_{\ell \neq i} g_i(A_{\ell} \cap R_i) + \sum_{i \in I
                \cap (\bar{J} - \{1\})} n_i.$$
\end{corollary}
\begin{proof}
        We have
        \begin{align*}
                \sum_{i \in \bar{J} - \{1\}} g_i(A'_1 \cap R_i) &\leq \sum_{i
                \in \bar{J} - \{1\}} g_i(A_1 \cap R_i) + \sum_{i \in \bar{J}
                - \{1\}} (u_i - u_{i1})  + \sum_{i \in I \cap (\bar{J} -
                \{1\})} n_i\\
                \sum_{i \in \bar{J} - \{1\}} \sum_{\ell \notin \{i, 1\}}
                g_i(A'_{\ell} \cap R_i) &\leq \sum_{i \in \bar{J} - \{1\}}
                \sum_{\ell \notin \{i, 1\}} g_i(A_{\ell} \cap R_i) - \sum_{i
                \in \bar{J} - \{1\}} \sum_{\ell \notin \{i, 1\}} u_{i\ell}\\
                &= \sum_{i \in \bar{J} - \{1\}} \sum_{\ell \notin \{i, 1\}}
                g_i(A_{\ell} \cap R_i) - \sum_{i \in \bar{J} - \{1\}} (u_i -
                u_{i1})\\
                \sum_{i \in \bar{J} - \{1\}} \sum_{\ell \neq i} g_i(A'_{\ell}
                \cap R_i) &= \sum_{i \in \bar{J} - \{1\}} g_i(A'_1 \cap R_i)
                + \sum_{i \in \bar{J} - \{1\}} \sum_{\ell \notin \{i, 1\}}
                g_i(A'_{\ell} \cap R_i)\\
                &= \sum_{i \in \bar{J} - \{1\}} \sum_{\ell \neq i}
                g_i(A_{\ell} \cap R_i) + \sum_{i \in I \cap (\bar{J} -
                \{1\})} n_i.
        \end{align*}
\end{proof}

\begin{claim}
        We have
                $$\sum_{\ell \neq 1} g_1(A'_{\ell} \cap R_1) \leq k.$$
\end{claim}
\begin{proof}
        Consider an $\ell \neq 1$. Note that we may assume that $1 \notin
        J$, since otherwise $A'_{\ell} \cap R_1 = \emptyset$ and
        $g_1(A'_{\ell} \cap R_1) = \phi(|A'_{\ell} \cap R_1|) = 0$.
        Since at most one element of $A_{\ell} \cap R_1$ is happy and we
        reassigned all unhappy elements of $R_1$ to terminal $1$, we have
        $|A'_{\ell} \cap R_1| \leq 1$ and thus $g_1(A'_{\ell} \cap R_1) =
        \phi(|A'_{\ell} \cap R_1|) \leq 1$.
\end{proof}

\noindent
By symmetry, the cost over columns is the same as the cost over rows.
It follows that the cost $C'$ of the resulting assignment satisfies
\begin{align*}
        C' &= \sum_{\ell = 1}^k f(A'_{\ell}) = 2\sum_{i = 1}^k \sum_{\ell
        = 1}^k g_i(A'_{\ell} \cap R_i)\\
        &= 2 \left(\sum_{i = 1}^k g_i(A'_i \cap R_i) + \sum_{\ell \neq 1}
        g_1(A'_{\ell} \cap R_1) + \sum_{i \in J - \{1\}} \sum_{\ell \neq
        i} g_i(A'_{\ell} \cap R_i) + \sum_{i \in \bar{J} - \{1\}}
        \sum_{\ell \neq i} g_i(A'_{\ell} \cap R_i) \right)\\
        &\leq 2\left(\sum_{i = 1}^k \sum_{\ell = 1}^k g_i(A_{\ell} \cap
        R_i) - \sum_{i \in I} n_i - \sum_{i \in J - \{1\}} n_i + \sum_{i
        \in I \cap (\bar{J} - \{1\})} n_i + 3k\right)\\
        &\leq 2\left(\sum_{i = 1}^k \sum_{\ell = 1}^k g_i(A_{\ell} \cap
        R_i) + 3k\right)\\
        &= C + 6k.
\end{align*}

\subsubsection{Proof of Lemma~\ref{lem:cost-symmetric-nice}}
Consider a nice symmetric assignment. We define a coloring of the complete
graph $K_k$ as follows. Vertex $i$ represent terminal $(i,i)$ and the
edge $(i, j)$ represents the pair $(i, j)$, where $i \neq j$. We use
the symmetric assignment to define a coloring of the edges of $K_k$:
Consider a pair $(i, j)$ where $i \neq j$, and suppose $(i, j)$ is
assigned to terminal $\ell$; we assign color $\ell$ to the edge $(i,
j)$ of $K_k$. (Since the assignment is symmetric, $(i, j)$ and $(j,
i)$ are both assigned to the same terminal.) Let $E_{\ell}$ denote the
set of all edges of $K_k$ that have color $\ell$, and let $N_i$
denote the set of all edges incident to vertex $i$.
We say that an edge $(i, j)$ is \emph{unhappy} if the element $(i,
j)$ is unhappy. Let $\mu$ denote the number of unhappy edges.

We partition $[k]$ into $L$ and $\bar{L} = [k] - L$, where $L$ is the
set of all vertices $i$ such that $|E_i \cap N_i| < \gamma / 2 - 1$.
Note that since the assignment is nice, we actually have $|E_i \cap N_i| = 0$
in this case.
Note that vertices $i \in L$ are ``lonely'', i.e., none of the edges
incident to $i$ have color $i$. Moreover, $1 \notin L$, since vertex
$1$ is incident to at least $\gamma/2 - 1$ edges of color $1$.
We further partition $L$ into $L_1$ and $L_2 = L - L_1$, where $L_1$
is the set of all vertices $i \in L$ such that all the edges incident
to $i$ have color $1$.

Let $a = |L_1|$, $b = |L_2|$, and $c = |\bar{L}| = k - a - b$. Let
$C$ be the cost of the assignment. Recall that $\mu$ is the number of
unhappy edges.

\begin{claim}
        We have
                $$C \geq 2k(k - 2) - 2(2a + b) \gamma + 4(a + b) + 4\mu.$$
\end{claim}
\begin{proof}
        Let $A_{\ell}$ denote the set of all elements that are assigned
        to terminal $\ell$. Let $R_i$ (resp. $C_i$) denote the $i$-th row
        (resp. the $i$-th column). Note that, since the assignment is
        symmetric, $g_i(A_{\ell} \cap R_i) = g_i(A_{\ell} \cap C_i)$.
        Therefore
        \begin{align*}
                {C \over 2} &= {1 \over 2} \sum_{\ell = 1}^k f(A_{\ell})
                = \sum_{\ell = 1}^k \sum_{i = 1}^k g_i(A_{\ell} \cap R_i)\\
                &= \sum_{i = 1}^k g_i(A_i \cap R_i) + \sum_{\ell \neq i}
                g_i(A_{\ell} \cap R_i)\\
                &= \sum_{i = 1}^k \phi(k - 1 - |E_i \cap N_i|) + \sum_{i =
                1}^k \sum_{\ell \neq i} \phi(|E_{\ell} \cap N_i|).
        \end{align*}
        We consider the two sums separately.
        \begin{align*}
                 \sum_{i = 1}^k \phi(k - 1 - |E_i \cap N_i|) &= \sum_{i \in
                 L} \phi(k - 1 - |E_i \cap N_i|) + \sum_{i \in \bar{L}}
                 \phi(k - 1 - |E_i \cap N_i|).
        \end{align*}
        For each $i \in L$, $|E_i \cap N_i| = 0$ and thus $\phi(k - 1 -
        |E_i \cap N_i|) = \phi(k - 1)$. For each $i \in \bar{L}$, $|E_i
        \cap N_i| \geq \gamma / 2$ and thus $\phi(k - 1 - |E_i \cap N_i|)
        = k - 1 - |E_i \cap N_i|$. Thus
        \begin{align*}
                \sum_{i = 1}^k \phi(k - 1 - |E_i \cap N_i|) &= \phi(k - 1)
                |L| + (k - 1)|\bar{L}| - \sum_{i \in \bar{L}} |E_i \cap
                N_i|\\
                &= k(k - 1)  - (\gamma - 2)(a + b) - \sum_{i \in \bar{L}}
                |E_i \cap N_i|.
        \end{align*}
        Now consider the second sum. Consider an index $i \in L_1$ and an
        index $\ell \neq i$. If $\ell = 1$, we have $|A_{\ell} \cap N_i|
        = k - 1$ and, if $\ell \neq 1$, we have have $|A_{\ell} \cap N_i|
        = 0$. Since $1 \notin L_1$, we have
                $$\sum_{i \in L_1} \sum_{\ell \neq i} \phi(|E_{\ell} \cap
                N_i|)  = \sum_{i \in L_1} \phi(|E_1 \cap N_i|) = \phi(k - 1)
                |L_1|.$$
        Therefore we have
        \begin{align*}
                \sum_{i = 1}^k \sum_{\ell \neq i} \phi(|E_{\ell} \cap N_i|)
                &= \sum_{i \in L_1} \sum_{\ell \neq i} \phi(|E_{\ell} \cap
                N_i|) + \sum_{i \in L_2} \sum_{\ell \neq i} \phi(|E_{\ell}
                \cap N_i|) + \sum_{i \in \overline{L}} \sum_{\ell \neq i}
                \phi(|E_{\ell} \cap N_i|)\\
                &= \phi(k - 1)|L_1| + \sum_{i \in L_2} \sum_{\ell \neq i}
                \phi(|E_{\ell} \cap N_i|) + \sum_{i \in \bar{L}} \sum_{\ell
                \neq i} |E_{\ell} \cap N_i|\\
                &= (k + 1 - \gamma) a + \sum_{i \in L_2} \sum_{\ell \neq i}
                \phi(|E_{\ell} \cap N_i|) + \sum_{i \in \bar{L}} \sum_{\ell
                \neq i} |E_{\ell} \cap N_i|.
        \end{align*}
        Consider an index $i \in L_2 \setminus \{1\}$. If there exists an $\ell
        \neq i$ such that $|E_{\ell} \cap N_i| > k - \gamma/2$, at least
        $k - \gamma / 2$ edges incident to $i$ are unhappy and therefore
        all of the edges incident to $i$ have color $1$. But $i$ has at
        least $\gamma / 2$ edges incident to it that have color $i$ and
        therefore $|E_{\ell} \cap N_i| \leq k - \gamma / 2$ for all $\ell
        \neq i$. Thus, for each $i \in L_2 \setminus  \{1\}$, we have
                $$\sum_{\ell \neq i} \phi(|E_{\ell} \cap N_i|) = \sum_{\ell
                \neq i} |E_{\ell} \cap N_i|.$$
        Therefore
                $$\sum_{i = 1}^k \sum_{\ell \neq i} \phi(|E_{\ell} \cap N_i|)
                \geq (k + 1 - \gamma) a - k + \sum_{i \in L_2 \cup
                \bar{L}} \sum_{\ell \neq i} |E_{\ell} \cap N_i|.$$
        By combining the inequalities derived so far, it follows that
        \begin{align*}
                {C \over 2} &\geq k(k - 1) + k(a - 1) - (2a + b)\gamma +
                3a + 2b - \sum_{i \in \bar{L}} |E_i \cap N_i| + \sum_{i \in
                L_2 \cup \bar{L}} \sum_{\ell \neq i} |E_{\ell} \cap N_i|.
        \end{align*}
        Recall that, for each $i \in L_1$, all of the edges incident to
        $i$ have color one. Thus $\sum_{i \in L_1} \sum_{\ell \neq i}
        |E_{\ell} \cap N_i| = \sum_{i \in L_1} |E_1 \cap N_i| = (k - 1)a$,
        and therefore
                $${C \over 2} \geq k(k - 2) - (2a + b) \gamma + 2(a + b)
                - \sum_{i \in \bar{L}} |E_i \cap N_i| + \sum_{i = 1}^k
                \sum_{\ell \neq i} |E_{\ell} \cap N_i|.$$
        In fact, the sum of $\bar{L}$ can be replaced by a sum over all $i$,
        because vertices in $L$ do not have any edge of color $i$ incident to them.
        Each unhappy edge contributes $2$ to the sum $\sum_{i =
        1}^k \sum_{\ell \neq i} |E_{\ell} \cap N_i|$, and each happy edge
        contributes $1$ to the sum $\sum_{i = 1}^k \sum_{\ell \neq i}
        |E_{\ell} \cap N_i|$. Moreover, each unhappy edge contributes $0$
        to the sum $\sum_{i = 1}^k |E_i \cap N_i|$, and each happy edge
        contributes $1$ to the sum $\sum_{i = 1}^k |E_i \cap N_i|$. Recall that the number
        of unhappy edges is $\mu$. Thus we
        have
                $$\sum_{i = 1}^k \sum_{\ell \neq i} |E_{\ell} \cap N_i| -
                \sum_{i \in \bar{L}} |E_i \cap N_i|
                = \sum_{i = 1}^k \sum_{\ell \neq i} |E_{\ell} \cap N_i| -
                \sum_{i = 1}^k |E_i \cap N_i|
                = 2\mu.$$
        Therefore
                $${C \over 2} \geq k(k - 2) - (2a + b) \gamma + 2(a + b)
                + 2\mu.$$
\end{proof}

\begin{claim}
        The number of unhappy edges $\mu$ satisfies
                $$\mu \geq {(a + b)(a + b - 1) \over 2} + a(k - a - b - 1).$$
\end{claim}
\begin{proof}
        For each vertex $i \in L$, there are no edges incident to $L$
        that have color $i$. Therefore all edges with both endpoints in
        $i$ are unhappy; there are $(a + b)(a + b - 1) / 2$ such edges.
        Additionally, the edges incident to a vertex in $L_1$ have color
        $1$. Therefore all edges $(i, j)$ such that $i \in L_1$ and $j
        \in \bar{L}$ and $j \neq 1$ are also unhappy; there are at least
        $a(c - 1) = a(k - a - b - 1)$ such edges.
\end{proof}

\begin{proof}[Proof of Lemma~\ref{lem:cost-symmetric-nice}]
        It follows from the previous claims that the total cost $C$
        satisfies
                $$C \geq 2k(k - 2) + 2a(2k - 2\gamma - a - 2) + 2b(b - \gamma).$$
        Let $h(a, b) = 2k(k - 2) + 2a(2k - 2\gamma - a - 2) + 2b(b -\gamma)$. 
        The function $h(a, b)$ is concave in $a$ and therefore
        it is minimized when $a = 0$ or $a = k$.

        Consider the case $a = 0$. By taking the derivative of $h(0, b)$
        with respect to $b$ and setting it to zero, we see that $h(0, b)$
        is minimized when $b = \gamma / 2$.  Consider the case $a = k$.
        Since $a + b \leq k$, we have $b = 0$.  Thus $C \geq
        \min\{h(0, \gamma / 2), h(k, 0)\}$. We choose $\gamma = 2 \lfloor 2k
         - \sqrt{3k^2 - 2k} \rfloor$ which makes the two expressions equal (up to the rounding of $\gamma$).
        This implies 
        $$C \geq (8\sqrt{3} - 12)k^2 - O(k).$$
\end{proof}

\subsection{From oracle hardness to computational hardness}
\label{sec:oracle-to-NP}

Here we explain how the hardness results in the value oracle model
(the first parts of Theorem~\ref{thm:SMP-hardness} and
\ref{thm:sym-SMP-oracle-hardness}) imply the analogous computational
hardness results in these theorems.  The key technique here is the
use of list-decodable codes for the encoding of  a ``hidden
partition", as introduced in \cite{DobVon12}.  The technique,
although presented in \cite{DobVon12} for submodular maximization
problems, applies practically without any change to submodular
minimization as well. For concreteness, let us summarize how this
technique applies to the $\submp$ problem:

We generate instances on a ground set $X = N \times V$, where $V = [k] \times [k]$ is the ground set of our initial symmetric instance. The goal is to present an objective function $f:2^X \rightarrow \RR_+$ in a form that makes it difficult to distinguish the elements
in sets $N \times \{(i,j)\}$ from elements in $N \times \{(j,i)\}$, for any fixed $i,j$. More precisely, we construct a partition of $N \times V$ into sets $\{A_{ij}: (i,j) \in [k] \times [k] \}$ such that it is hard to distinguish the elements in $A_{ij}$ from the elements in $A_{ji}$.

We accomplish this by way of a list-decodable code as follows (the construction is simplified for this particular case):
We consider a Unique-SAT instance $\phi$ with $m$ variables, and an encoding function $E:\{0,1\}^m \rightarrow \{0,1\}^{N \times {k \choose 2}}$. The satisfying assignment $x^*$ of $\phi$ (if it exists) implicitly defines a partition of $N \times [k] \times [k]$ into subsets
$\{A_{ij}: (i,j) \in [k] \times [k]\}$, by specifying that for $i<j$,
$$ A_{ij} = \{ (a,i,j) \in N \times [k] \times [k]: (E(x^*))_{a,\{i,j\}} = 0 \} \cup \{ (a,j,i) \in N \times [k] \times [k]: (E(x^*))_{a,\{i,j\}} = 1 \}.$$
In other words, the string $E(x^*)$ tells us which pairs of elements $(a,i,j), (a,j,i)$ should be swapped, to obtain the sets $A_{ij}$ from the sets $N \times \{(i,j)\}$. In particular, if $E(x^*)$ is the all-zeros string, then $A_{ij} = N \times \{(i,j)\}$.
For $i=j$, we define $A_{ii} = N \times \{(i,i)\}$ in every case.

The encoding of an objective function for $\submp$ then consists purely of the Unique-SAT instance $\phi$, which is interpreted as follows. If $\phi$ is satisfiable, it induces the sets $A_{ij}$ as above, and the objective function is understood to be $\hat{f}(S) = \hat{F}(\xi)$,
where $\xi_{ij} = \frac{|S \cap A_{ij}|}{|A_{ij}|}$ and $\hat{F}$ is the smooth submodular function we described in Section~\ref{sec:oracle-hardness}. If $\phi$ is not satisfiable, then the objective function is understood to be $\hat{f}(S) = \hat{G}(\xi)$, with the same notation as in Section~\ref{sec:oracle-hardness}. We know that the gap between the optima in these two cases can be made arbitrarily close to $2-2/k$,
and if we could distinguish the two cases, we could determine whether a Unique-SAT instance has 0 or 1 satisfying assignments, which would imply $NP=RP$.

The question, however, is whether this is a legitimate encoding of the objective function, in the sense that given the encoding $\phi$ and a set $S$, the value of $f(S)$ can be computed efficiently. The main insight of \cite{DobVon12} is that list-decodable codes allow us to evaluate such functions efficiently. The reason for this is that the functions $\hat{F}(x), \hat{G}(x)$ are constructed in such a way (going back to \cite{V09}, see Lemma~\ref{lemma:final-fix}) that their value depends on the hidden partition $A_{ij}$ only if $x$ is significantly unbalanced, meaning that $\sum_{a \in A_{ij}} x_a - \sum_{a \in A_{ji}} x_a > \beta n$ for some $i,j \in [k]$ and constant $\beta > 0$.
Equivalently, this means that to evaluate $\hat{f}(S)$ and $\hat{g}(S)$, we need to determine the hidden partition only if
$|S \cap A_{ij}| - |S \cap A_{ji}| > \beta n$ for some $i,j \in [k]$. Fortunately, there are known list-decodable codes that allow us
to determine the satisfying assignment $x^*$ and the induced partition $\{A_{ij}: i,j \in [k] \}$, given any such set $S$ (because it corresponds to a corrupted message which is relatively close to the codeword $E(x^*)$).

To summarize, for a given set $S$, we are either able to determine $x^*$ and the induced partition $\{A_{ij}: i,j \in [k] \}$, in which case we are obviously able to evaluate $\hat{f}(S)$ or $\hat{g}(S)$; or we conclude that $S$ does not satisfy $|S \cap A_{ij}| - |S \cap A_{ji}| > \beta n$ for any $i,j \in [k]$, and in this case we do not need to know the partition $\{A_{ij}: i,j \in [k] \}$ because $\hat{f}(S) = \hat{g}(S)$ and $\hat{g}$ does not depend on the partition of $A_{ij} \cup A_{ji}$ into $A_{ij}$ and $A_{ji}$. This concludes the proof that we presented a legitimate encoding of the objective function.

For further details on this technique to convert oracle hardness results into computational hardness, we refer the reader to \cite{DobVon12}.

\end{document}